%% file: main.tex
\documentclass{article}
    
    \usepackage{amsmath,amssymb,amsthm,mathtools}
    \usepackage{hyperref}
    \usepackage{xspace}
    \usepackage{enumitem}
    \usepackage{microtype}
    \usepackage{booktabs}
    \usepackage{authblk}
    \usepackage{url}
    \input{preamble}

    \newcommand{\E}{\mathbb{E}}
    \newcommand{\Prb}{\mathbb{P}}
    \newcommand{\1}{\mathbf{1}}
    \newcommand{\ceil}[1]{\left\lceil #1\right\rceil}
    
    \newcommand{\id}{\mathsf{id}}
    \newcommand{\lane}{\mathsf{lane}}
    \newcommand{\sym}{\ell_{\mathrm{sym}}}
    \newcommand{\Mh}{M_h}
    \newcommand{\Ms}{M_s}
    \newcommand{\KL}{D}
    \newcommand{\hash}{\mathsf{H}}
    \newcommand{\Sig}{\mathsf{Sig}}
    
    \newcommand{\addr}{\mathsf{addr}}
    \newcommand{\ticket}{\mathsf{ticket}}
    \newcommand{\nonce}{\mathsf{nonce}}
    
    \theoremstyle{definition}
    \newtheorem{assumption}{Assumption}
    
    \theoremstyle{plain}
    \newtheorem{theorem}{Theorem}
    \newtheorem{corollary}{Corollary}
    
    \title{A mechanism design overview of Sedna}
    \author[1,2]{Benjamin Marsh\thanks{ben@seinetwork.io}}
    \author[1]{Alejandro Ranchal-Pedrosa}
    \affil[1]{Sei Labs}
    \affil[2]{University of Portsmouth}
    
\begin{document}
    \maketitle
    
\begin{abstract}
Sedna is a coded multi-proposer consensus protocol in which a sender shards a transaction payload into rateless symbols and disseminates them across parallel proposer lanes, providing high throughput and ``until decode'' privacy. This paper studies a sharp incentive failure in such systems. A cartel of lane proposers can withhold the bundles addressed to its lanes, slowing the chain's symbol accumulation while privately pooling the missing symbols. Because finalized symbols become public, the cartel's multi-slot information lead is governed by a chain level delay event where the chain fails to accumulate the $\kappa$ bundles needed for decoding by the honest horizon $t^\star=\lceil \kappa/m\rceil$. We characterize the resulting delay probability with KL-type large deviation bounds and show a knife edge pathology when the slack $\Delta=t^\star m-\kappa$ is zero such that withholding a single bundle suffices to push inclusion into the next slot with high probability.
    
We propose \textsf{PIVOT-$K$}, a Sedna native pivotal bundle bounty that concentrates rewards on the $\kappa$ bundles that actually trigger decoding, and we derive explicit incentive compatibility conditions against partial and coalition deviations. We further show that an adaptive sender ``ratchet'' that excludes lanes whose tickets were not redeemed collapses multi-slot withholding into a first slot deficit when $t^\star\ge 2$, reducing the required bounty by orders of magnitude. We close by bounding irreducible within slot decode races and providing parameter guidance and numerical illustrations. Our results show that for realistic parameters Sedna can reduce MEV costs to 0.04\% of the transaction value.
\end{abstract}

\section{Introduction}
\label{sec:intro}
Leader based consensus protocols concentrate proposal power in a single slot leader, this creates well known incentive and governance problems around censorship and MEV extraction. A growing design trend is multi-proposer consensus, where multiple proposers finalize in parallel each proposal (e.g., as independent proposal ``lanes''), increasing throughput and potentially reducing single leader market power. Sedna is a concrete instantiation of this trend using coded dissemination. A sender splits a payload into rateless symbols and, in each slot, contacts $m$ of the $n$ lanes and transmits a small
bundle of $s$ symbols to each contacted lane. Once the chain has accumulated $K$ verified symbols (equivalently $\kappa=\lceil K/s\rceil$ bundles), the payload can be decoded and executed. The Sedna protocol paper~\cite{sedna2025} studies correctness, throughput, and ``until-decode'' privacy, we focus in this work instead on incentives asking do lane proposers actually want to include the bundles they receive (equivalently:  what is the MEV protection cost in Sedna)? A cartel controlling a fraction $\beta$ of lanes can receive its addressed bundles but strategically withhold them from finalization. Because finalized symbols are publicly observable after each slot, the cartel's information set is all finalized symbols plus any withheld cartel symbols. Therefore, the cartel obtains a multi-slot information lead exactly when the chain is delayed past the honest inclusion horizon. Under full inclusion, the chain deterministically reaches the decode threshold after
\[
    t^\star := \left\lceil \frac{\kappa}{m}\right\rceil
    \quad\text{slots, with slack}\quad
    \Delta := t^\star m-\kappa \in \{0,1,\dots,m-1\}.
\]
Let $W_{t^\star}$ be the number of cartel addressed bundles withheld before $t^\star$. The corrected lead event reduces to a one line condition:
\[
    T>t^\star \quad\Longleftrightarrow\quad W_{t^\star}>\Delta.
\]
This makes the knife edge immediate: if $\Delta=0$ (equivalently $m\mid \kappa$), one withheld bundle in the first $t^\star$ slots forces $T>t^\star$ whenever the cartel is contacted at least once.

The paper makes six contributions. First, it identifies the threat event where the cartel's multi-slot informational advantage is exactly a chain delay event beyond the honest horizon $t^\star=\lceil \kappa/m\rceil$, governed by the slack $\Delta=t^\star m-\kappa$. Second, under the static sender and the hypergeometric contact model, it derives KL-type large deviation bounds for the delay probability and shows the resulting sawtooth dependence on $\kappa \bmod m$, with a severe knife edge at $\Delta=0$. Third, it establishes a pathwise dominance result: full withholding maximizes the delay probability over all dynamic cartel policies, so the worst case MEV threat is always governed by the $w=0$ baseline (Theorem~\ref{thm:worst_w0_rigorous}). Fourth, it proposes \textsf{PIVOT-$K$}, a Sedna native pivotal bundle bounty that pays only the decoding relevant prefix of bundles and therefore concentrates incentives where they matter, and proves its minimax optimality among all budget capped rank based rules. Fifth, it decomposes the monolithic cartel into individually rational lane operators and shows that for $\Delta>0$, unilateral withholding is already strictly unprofitable as any successful delay requires coordinating at least $\Delta+1$ withheld bundles, and the aggregate direct revenue sacrifice of any such coalition exceeds the total extractable MEV from this transaction. Sixth, it shows that an adaptive sender ratchet collapses multi-slot withholding into a first slot deficit when $t^\star\ge 2$, and it complements that mitigation result with bounds on the residual within slot race that no finalization based payment rule can remove. For realistic parameters, we show that Sedna can provide MEV protection at a 0.04\% cost relative to the transaction value.
    
    \begin{table}[t]
    \centering
    \caption{Core notation.}
    \label{tab:notation}
    \small
    \begin{tabular}{@{}p{0.22\linewidth}p{0.72\linewidth}@{}}
    \toprule
    \textbf{Symbol} & \textbf{Meaning} \\
    \midrule
    $n$ & lanes per slot \\
    $m$ & lanes contacted per slot by the sender \\
    $s$ & symbols per bundle (per contacted lane) \\
    $K$ & decode threshold in symbols \\
    $\kappa=\lceil K/s\rceil$ & decode threshold in bundles \\
    $t^\star=\lceil \kappa/m\rceil$ & honest inclusion horizon (slots) \\
    $\Delta=t^\star m-\kappa$ & slack at the honest horizon \\
    $r = m - \Delta$ & bundles still needed at the start of slot $t^\star$ \\
    $\beta$ & cartel lane fraction \\
    $A_t$ & cartel lanes contacted in slot $t$ \\
    $S_t=\sum_{u\le t} A_u$ & cumulative cartel contacts \\
    $W_{t^\star}$ & withheld cartel bundles up to $t^\star$ \\
    $\phi$ & proposer share of the per-bundle bandwidth fee \\
    $f = \phi\,c\,L(s)$ & proposer fee per included bundle \\
    $B$ & sender posted bounty budget (paid at decode) \\
    $\pi(w)$ & cartel's expected inclusion share under policy $w$ \\
    $\gamma$ & per-slot discount factor \\
    $\alpha v$ & sender value at risk from MEV extraction \\
    $q_0$ & knife edge delay probability $\Prb[S_{t^\star}>0]$ \\
    $\rho(a,r)$ & conditional within slot decode success probability \\
    $\bar\rho$ & worst case $\rho$ over feasible $(a,r)$ \\
    \bottomrule
    \end{tabular}
    \end{table}
    
Section~\ref{sec:model} states the protocol abstraction, sender strategy, and adversarial model. Section~\ref{sec:threat} formalizes the withholding attack and derives delay probability bounds. Section~\ref{sec:mech} introduces PIVOT-K and motivates its pivotal prefix structure. Section~\ref{sec:incentives} derives incentive compatibility conditions (including coalition deviations). Section~\ref{sec:ratchet} analyzes adaptive sender policies (the ratchet), and Section~\ref{sec:micro} bounds within slot decode races. Section~\ref{sec:viability} assembles the equilibrium viability conditions, Section~\ref{sec:numerics} provides parameter guidance and evaluations, and Section~\ref{sec:conclusion} concludes.
    
\section{Related work}
\label{sec:related}
This paper is closest to the recent line of work on multiple concurrent proposers, such as Sei Giga~\cite{marsh2025sei}, and coded dissemination. Sedna itself is introduced in~\cite{sedna2025}. Garimidi et al.~\cite{garimidi2025mcp} propose MCP as a consensus design targeting selective censorship resistance and hiding, while Landers and Marsh~\cite{landers2025mcpmev} analyze MEV in MCP blockchains where multiple blocks become available before the final execution order is fixed. Our contribution is orthogonal to the existing work as we study the sender side dissemination layer in a
coded MCP system and identify a withholding attack whose geometry is controlled by the slack parameter $\Delta=t^\star m-\kappa$. \textsf{PIVOT-K}'s pivotal-bundle bounty is a form of baiting strategy~\cite{ranchalpedrosa2022trap}, making inclusion strictly dominant over withholding by concentrating a sender-funded reward on the lanes that trigger decoding rather than relying on slashable validator deposits.

A second closely related line studies fair ordering and fairness in replicated systems and DAG based ledgers. Ramseyer and Goel~\cite{ramseyer2024fairordering} formulate fair ordering through
streaming social choice, Travelers~\cite{xue2024travelers} proposes a communication efficient BFT fair ordering protocol, and the recent SoK on consensus for fair message ordering~\cite{li2024fairordering_sok} summarizes the broader design space. In the DAG setting,
Raikwar et al.~\cite{raikwar2023fairnessdag} survey fairness notions for DAG based DLTs, and Mah\'e and Tucci-Piergiovanni~\cite{mahe2025orderfairnessdag} evaluate order fairness violations under adversarial manipulation. Sedna's threat is adjacent to this literature but not identical to it, the strategic variable here is selective inclusion against a decode threshold rather than manipulation of a finalized total order.

A third line concerns MEV mitigation and private mempools. Daian et al.~\cite{flashbots2020} systematized MEV, and later work studies fee mechanisms, PBS, and block building incentives
\cite{roughgarden2021,gupta2023pbs,wu2023biddingwars}. Private order flow and accountable mempool mechanisms such as MEV-Share and L\O{}~\cite{mevshare,nasrulin2023lo} shift or audit the pre-confirmation information flow rather than altering dissemination geometry. On the privacy side, Shutter, Penumbra, and F3B are natural protocol comparators~\cite{shutter,penumbra,f3b}. Choudhuri et
al.~\cite{choudhuri2024bte} and Bormet et al.~\cite{bormet2025beatmev} show how batched threshold encryption can make mempool privacy practical. Rondelet and Kilbourn~\cite{rondelet2023mempoolprivacy} are especially relevant methodologically because they argue that mempool privacy must be analyzed through incentives as well as cryptography.
Our paper sits exactly in that gap: Sedna's until decode privacy is combinatorial rather than threshold encrypted, and the question is whether the intermediaries who receive coded fragments want to reveal them.

Finally, the dissemination layer connects to coded data availability and information dispersal. The original AVID protocol of Cachin et al.~\cite{cachin2004avid}, the near-optimal AVID construction of
Alhaddad et al.~\cite{alhaddad2022avid}, and DispersedLedger/AVID-M~\cite{yang2021dispersedledger} study how to disperse and retrieve coded data efficiently in Byzantine environments. Fino~\cite{malkhi2022fino} is particularly close in spirit because it combines DAG transport with MEV-aware hiding and explicitly discusses alternatives based on threshold techniques and AVID style dispersal. Incentive-compatible data availability sampling~\cite{buterin2024danksharding,hallandersen2024codedmerkle} is also conceptually adjacent: validators are paid to attest to data availability, analogous to lanes being paid to include bundles, though that literature focuses on sampling guarantees rather than decode-threshold races. Our paper differs from these works by making the intermediaries strategic and by analyzing a rank based reward rule over the decoding prefix rather than availability or communication complexity alone.

\section{Model}
\label{sec:model}

We formalise the system model: protocol primitives, sender
parameters, and adversarial model.

\subsection{System primitives and Sedna semantics}

There are $n$~proposer lanes per slot.  Time proceeds in discrete
slots $t = 1, 2, \dots$.  A Sedna transaction consists of a public
header and a payload~$x$.  The sender ratelessly encodes~$x$ into a
stream of symbols $\{\sigma_j\}_{j \ge 1}$ with decode threshold~$K$ where any $K$~distinct symbols suffice to reconstruct~$x$.

A bundle is the unit of lane-local inclusion.  When the
sender contacts lane~$\ell$ in slot~$t$, it constructs a bundle
$b = (\mathrm{body}(b),\, \ticket(b))$ carrying $s$~fresh symbol
indices and symbols, together with a lane/slot-bound fee ticket
(defined below).  Sedna resolves per-index deduplication via a
deterministic symbol resolution order, the lexicographic
order of $(t, \ell, \hash(\ticket(b)))$.  This order depends only on
on-chain data and sender-chosen randomness, so a lane proposer can
affect outcomes only by choosing whether to include admissible
bundles.  Inclusion occurs at the first slot where $K$~distinct
admissible symbols have been finalised, after which the payload is
executed in deterministic order.

Each lane~$\ell \in \{1, \dots, n\}$ has a protocol fixed reward address $\addr(\ell)$ for the current epoch. Whenever the sender contacts lane~$\ell$ in slot~$t$ and forms a bundle~$b$, it attaches a fee ticket signing the hash of the body:
\[
  \ticket(b)
  \;:=\;
  \Sig_{\mathrm{sender}}\!\big(
    \id,\; t,\; \ell,\; \hash(\mathrm{body}(b)),\;
    p,\; \addr(\ell),\; \nonce
  \big),
\]
where $\nonce$ is unique per issued ticket and $p > 0$ is the
per-bundle bandwidth fee.  A ticket is redeemable at most once and
is lane bound such that it can be redeemed only by including
bundle~$b$ in lane~$\ell$ in slot~$t$ (or within an explicitly
specified validity window), and only if the lane's reward address
equals $\addr(\ell)$.  A bundle occurrence is admissible iff
it carries a valid, unspent ticket satisfying these checks.  Only
admissible occurrences are charged/paid, counted for per-index
resolution and decoding, and eligible for bounty. In particular, a
proposer cannot profit by copying another lane's symbols such without a
fresh valid lane bound ticket, the copy is non-admissible and is
ignored.

Each admissible included bundle pays a fixed proposer fee
$f = \phi\,c\,L(s)$ independent of slot~$t$, the inclusion
time~$T$, and contemporaneous load. Including an additional
admissible bundle cannot reduce the proposer's discounted fee
revenue through fee market externalities.

\subsection{Sender dissemination}
\label{sec:sender}

Fix a transaction with decode threshold~$K$ and sender parameters
\[
  m \in \{1, \dots, n\}
  \quad\text{(lanes contacted per slot)},
  \qquad
  s \in \mathbb{Z}_{\ge 1}
  \quad\text{(symbols per contacted lane per slot)}.
\]
In each slot~$t$, the sender samples $m$~lanes uniformly without
replacement and sends one bundle to each.  Define
\[
  \kappa := \ceil{K/s}
  \quad\text{(bundles needed to reach $K$ symbols)}.
\]
Throughout we adopt $K = \kappa s$ (exact divisibility),
Appendix~\ref{app:ceilings} treats ceiling effects and shows all
additive error terms remain $O(m/\kappa)$.  More generally, a sender
can use a deterministic schedule $\{m_t\}_{t \ge 1}$ with
$m_t \in \{0, 1, \dots, n\}$ and fixed~$s$; the time varying
formulation is recorded in Appendix~\ref{sec:adaptive_sender} and
invoked in the adaptive sender analysis
(Section~\ref{sec:ratchet}).

\subsection{Adversarial model}
\label{sec:partial}

A fraction $\beta \in (0, 1)$ of the $n$~lanes is controlled by a cartel.  We parameterise deviations by an inclusion rate $w \in [0, 1]$: in each slot the cartel includes a fraction~$w$ of its received bundles and withholds the rest, so $w=1$ is full inclusion and $w=0$ is full withholding.  This family captures the key trade off that withholding slows on-chain accumulation (delaying the transaction, forgoing fee and bounty revenue) but may create an MEV extraction opportunity if the cartel privately decodes before the chain does. For tractability, Sections~\ref{sec:threat} and~\ref{sec:IC} use a fluid stationary model in which a fraction $w$ of cartel addressed bundles is included on average. The exact pathwise delay geometry is stated later for integer withholding policies via $X_t^\pi$ and $W_{t^\star}^\pi$
(Sections~\ref{sec:reduction} and~\ref{sec:dynamic_sabotage}). Accordingly, statements phrased in terms of $q^{(w)}$ should be read as properties of this fluid model rather than exact finite sample
identities for discrete bundle decisions.

Let $A_t$ be the number of cartel lanes contacted in slot~$t$:
\[
  A_t \sim \mathrm{Hypergeom}(n,\, \beta n,\, m),
  \qquad
  \{A_t\}_{t \ge 1}\ \text{i.i.d.}
\]
Write $H_t := m - A_t$ for the honest contacts, and define cumulative cartel contacts $S_t := \sum_{u=1}^t A_u$.

Under inclusion rate~$w$, the on-chain bundle count after slot~$t$ is
\[
  U_t(w)
  \;:=\;
  \sum_{u=1}^t \bigl(H_u + w A_u\bigr)
  \;=\;
  tm - (1 - w) S_t,
\]
and the on-chain inclusion time is
$T(w) := \inf\{t : U_t(w) \ge \kappa\}$.

Under full inclusion ($w = 1$), $U_t(1) = tm$, so inclusion is
deterministic at
\[
  t^\star := T(1) = \ceil{\kappa / m},
  \qquad
  \Delta := t^\star m - \kappa \in \{0, 1, \dots, m - 1\}.
\]
The slack~$\Delta$ measures redundancy, $t^\star m$ bundles are finalised by the honest horizon, but only $\kappa$ are needed for decoding. When $\Delta = 0$ (equivalently $m \mid \kappa$), full
inclusion yields exactly $\kappa$ bundles, a knife edge with zero slack.

A Sedna transaction progresses in two layers. Off-chain, the sender repeatedly samples $m$ lanes per slot and sends each contacted lane a bundle carrying $s$ fresh symbols together with a lane bound ticket. On-chain, only admissibly redeemed bundles count toward the decode threshold, so the public chain state after slot $t$ is the cumulative count of finalised admissible bundles rather than the sender's raw transmission history. A withholding cartel therefore affects the transaction only by creating a public bundle deficit relative to the honest benchmark. Under full inclusion the chain reaches $\kappa$ bundles by $t^\star$ deterministically, under deviation, every withheld cartel bundle subtracts one unit from that public total, and delay occurs exactly when the deficit exceeds the slack $\Delta$. This chain level viewpoint is the key reason the threat reduces to the event $W_{t^\star}>\Delta$ rather than to a cartel only decoding condition.

We model the cartel as a single actor parameterised by~$\beta$. This worst case abstraction assumes perfect coordination and costless MEV redistribution. In practice, the cartel comprises multiple independently motivated proposers, each making inclusion decisions on lane bound tickets as they are contacted. The MEV benefit from coordinated withholding requires at least $\Delta + 1$ withheld bundles, but under the static sender these withheld opportunities may come from repeated contacts to the same lane across different slots. Section \ref{sec:coalition} therefore phrases the local analysis in terms of withheld bundles rather than distinct cartel lanes.

Two timing questions govern the analysis, can the cartel gain a multi-slot information lead under full inclusion, and can it win a within slot race (decoding between receiving the final symbol and the slot's finalisation)? The first is resolved by the following assumption; the second is deferred to
Section~\ref{sec:micro}.

\begin{assumption}[Coarse slot timing]
\label{as:coarse}
Symbols delivered in slot~$t$ are not usable for private decoding
strictly before slot~$t$'s lane blocks are finalised.
\end{assumption}

Under Assumption~\ref{as:coarse} and full inclusion,
$T(1) = t^\star$ deterministically and
$q^{(1)} := \Prb[T(1) > t^\star] = 0$: full inclusion eliminates
any multi-slot information lead.  The residual within slot race
(Section~\ref{sec:micro}) is the irreducible privacy risk.  The
withholding driven delay probability
$q^{(w)} := \Prb[T(w) > t^\star]$ and its dependence on~$\Delta$
are characterised in Section~\ref{sec:threat}.  We assume coarse
slot timing throughout except in Section~\ref{sec:micro}.

\subsection{Strategic instruments and modeling scope}
\label{sec:instruments}
We gather here the key strategic instruments that make of Sedna an interesting tool to protect against MEV. The conservative direction of each choice is noted explicitly.

\paragraph{Sender instruments.}
\begin{itemize}[nosep]
\item \textbf{Hidden dissemination parameters.}
  The sender can keep $m$ (lanes contacted per slot), $s$ (symbols per
  bundle), and hence $\kappa = \lceil K/s \rceil$ privately. The cartel cannot infer $\kappa$ or the
  slack~$\Delta$ before committing to its withholding strategy.

\item \textbf{Bounty posting (\textsf{PIVOT-$K$}).}
  The sender can escrow a bounty~$B$ that is distributed exclusively
  among the first $\kappa$ included bundles upon decode
  (Section~\ref{sec:mech}).

\item \textbf{Fee setting.}
  Each bundle carries a lane bound fee ticket with per-bundle proposer
  revenue $f = \phi\,c\,L(s)$, set by the sender.

\item \textbf{Lane selection and adaptive exclusion (ratchet).}
  The sender chooses which lanes to contact each slot and can
  permanently exclude any lane whose ticket was not redeemed
  (Section~\ref{sec:ratchet}).
\end{itemize}

\paragraph{Adversary instruments.}
\begin{itemize}[nosep]
\item \textbf{Withholding or inclusion.}
  For each bundle addressed to a cartel lane, the cartel decides
  whether to include it on-chain (earning fee and potential bounty) or
  withhold it (forgoing direct revenue but potentially enabling a
  private decode).

\item \textbf{Observation of public state.}
  After each slot's finalization, all included bundles become public.
  The cartel's information set is therefore the union of all finalized
  symbols and any symbols from bundles it has privately withheld.
\end{itemize}

\paragraph{Cartel as independently motivated entities.}
The baseline analysis (Sections~\ref{sec:threat}--\ref{sec:IC}) models
the cartel as a single coordinated actor controlling a fraction~$\beta$
of lanes, assuming perfect internal coordination and costless MEV
redistribution.  This is the most pessimistic assumption from the
protocol's standpoint.  Section~\ref{sec:coalition} relaxes it by
decomposing the cartel into individually rational lane operators, each
making local inclusion decisions on lane bound tickets, and derives
coalition stability conditions under aggregate budget constraints.

\paragraph{Modeling choices favorable to Sedna.}
Two aspects of the model work in the protocol's favor and should be
noted.  First, Assumption~\ref{as:coarse} (coarse slot timing) rules
out within slot private decoding by fiat; Section~\ref{sec:micro}
relaxes this and bounds the residual risk.  Second, the fee market is
static: including an additional admissible bundle cannot reduce a
proposer's revenue through congestion externalities.  This guarantees
nonnegative marginal payoff from inclusion, which underpins the minimal
sabotage result (Theorem~\ref{thm:dyn_best_binary}).
Appendix~\ref{app:fee_market} records what changes when this
assumption fails.

    \section{The withholding problem}
    \label{sec:threat}
    
    We now characterise the economic threat that withholding creates.  A user
    has value $v > 0$ from inclusion/execution.  If the cartel decodes early
    and exploits the payload, the user loses $\alpha v$ for
    $\alpha \in [0,1]$.  Both user and cartel discount per slot by
    $\gamma \in (0,1)$.
    
    Each included bundle has length $L(s) := \Mh + s(\Ms + \sym)$~bytes, where $\Mh$ is the header size, $\Ms$ is the symbol metadata size, and $\sym$ is the per-symbol payload length.
    Let $c > 0$ be the per-byte price, so the total bandwidth fee for an
    admissible included bundle is $p := c\,L(s)$.  Let $\phi \in [0,1]$ be the
    share accruing to the lane proposer, giving per-bundle proposer revenue
    $f := \phi\,c\,L(s)$.
    
    Under honest behaviour ($w = 1$), the cartel earns fee revenue from
    included bundles; under withholding ($w < 1$), it forgoes some fee revenue
    but may gain a discounted MEV option
    $G(w) := \gamma^{t^\star} q^{(w)}$ from delaying on-chain
    inclusion past the honest horizon.  The severity of this option depends on
    the delay probability $q^{(w)} := \Prb[T(w) > t^\star]$ (the chance that the 
    cartel successfully forces the on-chain inclusion time to stretch beyond the 
    baseline horizon), which we now
    characterise. Section~\ref{sec:ratchet} shows that an adaptive sender can strictly
    reduce $q^{(w)}$ when $t^\star \ge 2$; the present section establishes the
    static sender baseline.

    Two modelling layers appear in this section.  Sections~\ref{sec:delay_geom} and~\ref{sec:largedev} first use the fluid stationary proxy parameterised by the inclusion rate~$w$ to expose the slack geometry and derive closed form KL-type delay bounds.  Section~\ref{sec:reduction} then returns to the exact pathwise model with integer inclusion decisions $X_t^\pi$ and shows that full withholding dominates every dynamic deviation for the purpose of maximizing delay.  The fluid model therefore supplies tractable delay exponents, while the worst case incentive compatibility (IC) reduction ultimately rests on an exact pathwise theorem.
    
    \subsection{Delay event geometry}
    \label{sec:delay_geom}
    To quantify the threat, we first translate the time based delay event 
    into a concrete bundle deficit condition. Specifically, we evaluate what it 
    takes mathematically for the cartel to prevent the chain from accumulating the 
    required $\kappa$ bundles by the honest horizon~$t^\star$. By mapping this 
    delay directly to the cartel's cumulative contacts and the system's built in 
    slack, we establish exactly how many bundles the cartel must withhold.
    \begin{lemma}[Equivalence of delay and bundle deficit]
    \label{lem:threshold_w}
    Within the fluid stationary $w$-model, the delay proxy $q^{(w)}$ is given by the event that the implied withheld mass exceeds the system's built in slack $\Delta$.
    \end{lemma}
    \begin{proof}
    By definition, a delay occurs ($T(w) > t^\star$) if and only if the on-chain 
    bundle count after $t^\star$ slots remains strictly below the threshold $\kappa$. 
    Using the cumulative count $U_{t^\star}(w) = t^\star m - (1-w)S_{t^\star}$ and 
    substituting $t^\star m = \kappa + \Delta$, we get:
    \[
      U_{t^\star}(w) < \kappa 
      \quad\Longleftrightarrow\quad 
      (\kappa + \Delta) - (1-w)S_{t^\star} < \kappa 
      \quad\Longleftrightarrow\quad 
      (1-w)S_{t^\star} > \Delta.
    \]
    \end{proof}
    
    \paragraph{Zero slack knife edge.}
    When $\Delta = 0$, the delay condition $(1-w)S_{t^\star} > 0$
reduces to $S_{t^\star} > 0$ for every $w<1$.  Thus delaying past
$t^\star$ requires only one cartel addressed bundle before the honest
horizon, and under the stationary $w$-model
\[
q^{(w)} = \Prb[S_{t^\star} > 0]
= 1 - \Prb[A_1 = 0]^{t^\star}
= 1 - \left(\frac{\binom{(1-\beta)n}{m}}{\binom{n}{m}}\right)^{t^\star}.
\]
Here we used independence across slots and the hypergeometric formula
for $\Prb[A_1=0]$.
    
    The geometry of the honest horizon $t^\star$ and its slack $\Delta$ is entirely 
    determined by how the decode threshold $\kappa$ fits into the contact rate $m$. 
    Whenever $\kappa$ is an exact multiple of $m$, the system lands on the zero slack 
    knife edge ($\Delta = 0$), meaning full honest inclusion yields exactly $\kappa$ 
    bundles by time $t^\star$. 
    
    For a latency sensitive sender targeting a single slot horizon ($t^\star = 1$), 
    they must configure $\kappa \le m$. The ``tight'' choice of $\kappa = m$ 
    yields zero slack, maximizing the risk of a multi-slot delay. To defend against 
    this, the sender can increase the slack ($\Delta = m - \kappa$) by expanding 
    the number of contacted lanes $m$, up to the absolute network maximum $n$. 
    
    However, this creates a fundamental operational trade off. Let $r = m - \Delta$ 
    denote the number of bundles needed in the final slot to cross the decode 
    threshold (detailed further in Section~\ref{sec:micro}). Increasing the slack 
    $\Delta$ strictly reduces the multi-slot delay risk ($q^{(w)}$). Yet, by widening 
    the broadcast net $m$, the sender also increases the expected number of cartel 
    contacts ($A_{t^\star}$). This makes the within slot race condition 
    ($A_{t^\star} \ge r$) much easier for the cartel to satisfy. Consequently, 
    operational tuning is a delicate balance: widening $m$ protects the transaction 
    from being delayed into future slots, but increases the risk that the cartel 
    intercepts enough symbols to decode and exploit the payload before the current 
    slot is sealed.
    
    \subsection{Large deviation bounds}
\label{sec:largedev}

Having established that a delay requires the cartel to intercept a specific
number of bundles, we now quantify how likely they are to actually receive them.
The cumulative cartel contact count $S_t = \sum_{u=1}^t A_u$ is a sum of
independent hypergeometric draws. The moment generating
function of each $A_u$ is pointwise dominated by that of $\mathrm{Bin}(m,\beta)$.
The standard Chernoff optimisation over the tilt parameter then yields the
binomial KL exponents for both tails (see e.g.~\cite{boucheron2013concentration}):
for any $\theta \in (0,1)$,
\begin{equation}
\label{eq:kl_tails}
\Prb\!\left[\frac{S_t}{tm} \ge \theta\right]
  \le \exp\{-tm\,\KL(\theta\|\beta)\}
  \;\;\text{for } \theta > \beta,
\qquad
\Prb\!\left[\frac{S_t}{tm} \le \theta\right]
  \le \exp\{-tm\,\KL(\theta\|\beta)\}
  \;\;\text{for } \theta < \beta.
\end{equation}

We now map these standard concentration bounds onto the delay mechanics of
Section~\ref{sec:delay_geom}.  Define the \emph{critical contact density},
the exact fraction of all transmitted bundles the cartel must intercept over
$t^\star$ slots to force a delay under inclusion rate $w$:
\[
  \theta_w := \frac{\Delta}{(1-w)t^\star m}.
\]
If $\theta_w \ge 1$, the delay event is impossible even under full
withholding, since the cartel cannot intercept enough bundle mass to
overcome the slack; in that case $q^{(w)}=0$.

\begin{theorem}[Delay probability under the fluid stationary model]
\label{thm:q_w}
Assume $\Delta>0$ and $\theta_w<1$. If $\theta_w\in(\beta,1)$, the cartel must receive
more than its expected contact share, and
\[
  q^{(w)} \le \exp\{-t^\star m\,\KL(\theta_w\|\beta)\}.
\]
If $\theta_w\in[0,\beta)$, the cartel's expected contact share already suffices
to induce delay, and the failure probability is exponentially small:
\[
  1-q^{(w)} \le \exp\{-t^\star m\,\KL(\theta_w\|\beta)\}.
\]
\end{theorem}

\begin{proof}
By Lemma~\ref{lem:threshold_w}, a delay occurs if
$(1-w)S_{t^\star} > \Delta$.  Dividing by $t^\star m$ translates this
condition to $S_{t^\star}/(t^\star m) > \theta_w$.  For $\theta_w<1$, the two
nontrivial regimes follow directly from the tail bounds~\eqref{eq:kl_tails};
for $\theta_w\ge 1$, impossibility is immediate since $S_{t^\star}\le t^\star m$.
\end{proof}

While the bounds above quantify the exponential decay of failure, the practical
implication of the low threshold case ($\theta_w < \beta$) is that a multi-slot
information lead is a near certainty.  To put this in perspective: on a knife
edge instance ($\Delta=0$) with a $20\%$ cartel ($\beta=0.2$) and an honest
horizon of one slot ($t^\star=1, m=20$), the exact delay probability under full
withholding approaches $q^{(0)} \approx 0.993$.  In this regime, the cartel
almost certainly succeeds in delaying the chain to extract MEV.

\subsection{Exact pathwise dominance over dynamic policies}
\label{sec:reduction}
    
    So far, we have analyzed stationary policies where the cartel includes a 
    fixed fraction $w$ of its bundles. In reality, a sophisticated cartel might 
    deploy a dynamic policy $\pi$, strategically choosing exactly which bundles 
    to include and which to withhold over time to optimize their net payoff. 
    
    However, if our primary concern is bounding the absolute wors case probability 
    of a multi-slot delay (the core driver of the MEV option), we do not need to 
    solve the cartel's complex dynamic optimization problem. Because every included 
    cartel bundle strictly pushes the chain closer to the decode threshold $\kappa$, 
    any policy that includes some bundles will reach the threshold no later 
    than a policy that includes none. 
    
    To formalize this, let $X_t^\pi \in \{0, 1, \dots, A_t\}$ denote the number of 
    cartel bundles included in slot $t$ under policy $\pi$. The total on-chain 
    bundle count by slot $t$ is then $U_t^\pi := tm - S_t + \sum_{u=1}^t X_u^\pi$. 
    We define the resulting on-chain inclusion time as $T^\pi := \inf\{t : U_t^\pi \ge \kappa\}$. 
    
    The probability that this dynamic policy successfully forces a multi-slot delay is:
    \[
    q^\pi := \Prb[T^\pi > t^\star].
    \]
    Because the MEV is extracted in the future, we discount it by $\gamma$ per slot. 
    The expected discounted MEV gain from this delay is therefore:
    \[
    G^\pi := \E\!\left[\gamma^{t^\star}\1\{T^\pi > t^\star\}\right] = \gamma^{t^\star}q^\pi.
    \]
    
    \begin{theorem}[Full withholding maximizes delay probability]
    \label{thm:worst_w0_rigorous}
    No matter what dynamic inclusion policy $\pi$ the cartel employs, the resulting 
    probability of forcing a multi-slot delay is bounded above by the delay 
    probability under full withholding ($w=0$). Consequently, the expected MEV 
    gain from causing a delay is also strictly bounded by the full withholding 
    baseline:
    \[
    q^\pi \le q^{(0)}, \qquad \text{and} \qquad G^\pi \le \gamma^{t^\star}\,q^{(0)}.
    \]
    \end{theorem}
    
    \begin{proof}
    Because the cartel can only add bundles to the chain ($X_u^\pi \ge 0$), the 
    on-chain bundle count under any policy $\pi$ satisfies $U_t^\pi \ge tm - S_t = U_t(0)$ 
    pathwise for all $t$. Thus, the hitting time to reach $\kappa$ bundles satisfies 
    $T^\pi \le T(0)$ pathwise. This implies the delay event $\{T^\pi > t^\star\}$ is 
    a subset of $\{T(0) > t^\star\}$, yielding $q^\pi \le q^{(0)}$. The bound on 
    $G^\pi$ follows directly from its definition.
    \end{proof}
    
Theorem~\ref{thm:worst_w0_rigorous} is the structural result that enables the later IC analysis to replace the intractable dynamic optimization with the clean worst case benchmark $q^{(0)}$. Appendix~\ref{app:dynamic} complements this exact reduction with a pathwise monotonicity lemma for \textsf{PIVOT-$K$} where adding cartel inclusions cannot reduce the cartel's pivotal prefix count. Under the present coarse model, where the multi-slot MEV term depends only on the binary event of delaying past $t^\star$, the dynamic best response therefore collapses to minimal sabotage (Theorem~\ref{thm:dyn_best_binary}). A full Markov Decision Process
analysis is only needed if the MEV benefit scales with the length of the information lead or if within slot timing is modeled more finely.

Since the delay probability is a structural property of the contact process and cannot be eliminated by parameter choice alone (the sawtooth pattern described in Section~\ref{sec:numerics} ensures every operating point trades multi-slot delay risk for within slot race risk), the natural response is to make withholding directly costly through a payment rule. The next section introduces such a rule.
    
    \section{The \textsf{PIVOT-$K$} mechanism}
    \label{sec:mech}
    
The delay analysis of Section~\ref{sec:threat} established that the cartel's MEV option $G(w) = \gamma^{t^\star} q^{(w)}$ can be non-negligible since on knife edge instances ($\Delta=0$), $q^{(w)}$ is exactly the probability of at least one cartel contact before $t^\star$ and can already be close to~$1$ for every $w<1$, even away from the knife edge the delay probability can be substantial when the cartel fraction~$\beta$ exceeds the critical contact density~$\theta_w$ (Theorem~\ref{thm:q_w}, case $\theta_w < \beta$).  Furthermore, full withholding maximises the delay probability over all dynamic policies (Theorem~\ref{thm:worst_w0_rigorous}), so the MEV threat is real and cannot be dismissed by assuming the cartel deviates only partially.
    
Transaction fees provide some deterrence, each withheld bundle forfeits proposer revenue~$f$, but this deterrent is blunt. Under honest inclusion, the chain accumulates $t^\star m$ bundles by the honest horizon, only $\kappa$ of these are actually needed for decoding.  The remaining $\Delta = t^\star m - \kappa$ bundles are redundant.  A lane proposer that withholds a redundant bundle sacrifices only~$f$, while potentially sharing in an MEV option worth $\alpha v \cdot \gamma^{t^\star}$. The core problem is that fees do not distinguish bundles that contribute to decoding from those that do not.
    
\textsf{PIVOT-$K$} addresses this by concentrating reward on the bundles that matter for decoding.  The sender escrows a bounty $B \ge 0$ alongside the transaction.  Upon inclusion, the protocol identifies the first $K$~distinct symbol indices used for decoding and distributes the bounty exclusively among the lanes that supplied them. 
    
Because each lane proposer cannot predict ex ante whether its bundle will fall inside the pivotal prefix, withholding reduces the cartel's expected bounty share. Pathwise, however, only the first $\kappa$ included bundles are rewarded: among the $\kappa+\Delta$ bundles present at the honest horizon, the last $\Delta$ positions are redundant for both decoding and bounty. Accordingly, a withheld bundle always sacrifices the fee~$f$, while bounty loss occurs only once withholding starts deleting positions from the pivotal prefix. If a coalition withholds $c$ bundles before $t^\star$, then at least $(c-\Delta)^+$ of the removed positions are pivotal, so its total direct revenue loss is at least
    \[
        c\,f + \frac{(c-\Delta)^+}{\kappa}B.
    \]
    For minimal sabotage $c=\Delta+1$, this lower bound becomes
    \[
    (\Delta+1)f + \frac{B}{\kappa}.
    \]

    Upon inclusion at height~$\tau$, the protocol scans finalised history
    in Sedna's deterministic symbol resolution order
    (\S\ref{sec:model}) and collects the first $K$~distinct symbol
    indices for transaction~$\id$.  This set is the pivotal index
    set~$\mathcal{I}^\star$.  For each $j \in \mathcal{I}^\star$,
    let $\lane(j)$ denote the lane whose admissible bundle first
    contributed index~$j$, this is well defined because admissibility
    requires a valid lane bound ticket, and the resolution order is
    deterministic. The \textsf{PIVOT-$K$} bounty
    rule\label{def:pivotK} pays $B/K$ to $\lane(j)$ for each
    $j \in \mathcal{I}^\star$. No bounty is paid outside
    $\mathcal{I}^\star$.
    
    Since admissibility requires a valid lane bound ticket, a proposer
    that replays or forwards another lane's symbol without a fresh ticket
    produces a non-admissible occurrence, which is ignored by both Sedna
    decoding and \textsf{PIVOT-$K$}.  Duplicate bundles therefore
    cannot inflate the payout, and ``bounty sniping'' is ruled out by
    the same ticket mechanism that enforces lane bound fees.
    
    Under the convention $K = \kappa s$ (Appendix~\ref{app:ceilings}
    treats the general case), each included bundle carries $s$~fresh
    indices. The first $\kappa$~included bundles therefore contribute
    exactly $K = \kappa s$ distinct indices and form
    $\mathcal{I}^\star$.  Each such bundle contains $s$~pivotal
    indices, each paying $B/K = B/(\kappa s)$, so the per-bundle
    bounty is $s \cdot B/(\kappa s) = B/\kappa$.  We adopt this per-pivotal bundle view throughout the incentive analysis, \textsf{PIVOT-$K$} pays $B/\kappa$ to each of the first $\kappa$ included bundles, regardless of which lane supplied them.
    
Most importantly for the incentive analysis, \textsf{PIVOT-$K$} increases the cartel's direct revenue loss when withholding deletes positions from the pivotal prefix. If a withheld bundle would have
been pivotal, the deviating lane sacrifices both the fee~$f$ and the bounty share $B/\kappa$ attached to that position; if the withheld bundle falls inside the redundant tail, it sacrifices only the fee.
Moreover, on any sample path, including an additional cartel bundle weakly increases the cartel's count among the first $\kappa$ included bundles 
(Appendix~\ref{app:dynamic}, Lemma~\ref{lem:pathwise_mono_pivot}), adding a cartel bundle to the prefix can only displace an honest bundle at the margin, never reduce the cartel's share.
    
This aggregate direct revenue floor is the quantity that matters in the coalition analysis. Section~\ref{sec:incentives} still uses the expected bounty share gap to derive conservative monolithic IC conditions. Section~\ref{sec:coalition} then exploits the sharper pathwise statement at the individual lane level, for $\Delta > 0$, withholding a single bundle causes no delay and 
produces no MEV option, so the deviator sacrifices the fee~$f$ (and 
possibly the bounty share $B/\kappa$) with no offsetting gain. For coordinated coalitions, the budget constraint~\eqref{eq:coalition_budget} compares the coalition's total direct revenue loss against the bounded MEV pie and yields a clean sufficient condition against delay inducing coordination. Finally, Section~\ref{sec:pivot_optimality} proves that among all budget capped rank-based pivotal rules, the uniform weights of \textsf{PIVOT-$K$} maximise the worst case bounty forfeiture from deleting any prescribed number of pivotal ranks (Theorem~\ref{thm:pivot_optimal}), and hence minimise the bounty budget~$B$ needed within this class.

The baseline mechanism is intentionally time neutral within the pivotal prefix such that every pivotal index is paid equally, independent of the slot in which it appears. A latency sensitive sender could refine this by multiplying the pivotal payment by a deterministic slot decay factor, for example paying $\delta^{t-1}B/K$ to pivotal indices first finalised in slot $t$ for some $\delta\in(0,1]$. Such a variant would reward earlier inclusion more aggressively, but it would also mix two distinct effects, rank in the decoding prefix and absolute time. We therefore keep the main mechanism uniform so that the incentive analysis isolates the pivotal prefix channel cleanly, and we regard time decaying pivotal rewards as a natural extension for deployments that care about latency as well as censorship resistance.

\section{Incentive analysis}
\label{sec:incentives}

With the delay characterisation (Section~\ref{sec:threat}) and the bounty
mechanism (Section~\ref{sec:mech}) in hand, we assemble the cartel's full
payoff and derive conditions under which full inclusion is the cartel's
best response, that is, conditions for \emph{incentive compatibility} (IC).
The cartel's objective under a stationary $w$-policy is
\[
  \Pi_{\mathcal{C}}(w)
  \;:=\;
  R_{\mathrm{fee}}(w)
  \;+\;
  R_{\mathrm{bounty}}(w)
  \;+\;
  \alpha v \cdot G(w),
\]
where $G(w) := \gamma^{t^\star} q^{(w)}$ is the discounted MEV option
from Section~\ref{sec:threat}.  Full inclusion ($w=1$) is incentive
compatible if $\Pi_{\mathcal{C}}(1) \ge \Pi_{\mathcal{C}}(w)$ for every
$w \in [0,1)$, meaning the cartel weakly prefers honest behaviour to
every stationary deviation.

An important informational asymmetry strengthens the IC analysis in
practice.  The cartel does not observe the sender's dissemination
parameters directly: the number of lanes contacted per slot~$m$, the
symbol count per bundle~$s$, and hence the decode threshold
$\kappa = \lceil K/s \rceil$ are all private to the sender.  Under
Sedna, the payload size is hidden behind a computationally hiding
commitment until decode, so the cartel cannot infer $\kappa$ (and
therefore the slack~$\Delta$) before committing to its first slot
withholding decision.  The analysis below grants the cartel full
knowledge of all parameters, any IC condition derived under this
assumption holds a fortiori when the cartel faces genuine uncertainty
about $\kappa$ and~$\Delta$.

The analysis proceeds in five steps.  Section~\ref{sec:payoffs} bounds
the two direct revenue components: fee revenue $R_{\mathrm{fee}}$ and
bounty revenue $R_{\mathrm{bounty}}$.  Section~\ref{sec:IC} combines
them with the delay probability bounds from Section~\ref{sec:threat} to
obtain an explicit sufficient IC condition.
Section~\ref{sec:pivot_optimality} shows that within the class of
budget capped rank based bounty rules, the uniform weights of
\textsf{PIVOT-$K$} are minimax optimal, minimising the bounty budget
needed for any given deterrence level.
Section~\ref{sec:dynamic_sabotage} analyses the cartel's dynamic best
response, showing that it reduces to minimal sabotage, withholding
exactly $\Delta+1$ bundles.  Finally,
Section~\ref{sec:coalition} recasts the monolithic threat in terms of
local one bundle deviations and aggregate coalition budget constraints,
showing that positive slack blocks any successful delay unless at least
$\Delta+1$ bundle opportunities are sacrificed.

\subsection{Payoff components}
\label{sec:payoffs}

The cartel's IC reduces to comparing the payoff under full inclusion
($w = 1$) with that under any deviation.  To make this comparison
tractable, we bound each of the three payoff components separately.
The fee and bounty terms represent the cartel's direct revenue, which
decreases with withholding; the MEV option is the indirect benefit,
which increases with withholding.  The IC condition asks when the
direct revenue loss exceeds the indirect gain.

\subsubsection{Fee revenue}
\label{sec:fee_revenue}

Fee revenue scales linearly in the inclusion rate~$w$ and the
cartel's expected contact share~$\beta$.  Under a stationary $w$-policy,
in slot $t$ the cartel includes $wA_t$ bundles and earns
$f \cdot wA_t$ (discounted by $\gamma^{t-1}$), giving
\[
  R_{\mathrm{fee}}(w)
  :=\E\!\left[\sum_{t=1}^{T(w)} \gamma^{t-1} f\cdot wA_t\right].
\]
Dropping the stopping time and using $\E[A_t]=\beta m$ yields the
upper bound
\[
  R_{\mathrm{fee}}(w)
  \le f\cdot \frac{w\beta m}{1-\gamma}.
\]
The fee gap $R_{\mathrm{fee}}(1) - R_{\mathrm{fee}}(w)$ is therefore
nonnegative and grows as $w$ decreases: the cartel sacrifices more fee
revenue the more it withholds.  This gap enters the IC condition
(Theorem~\ref{thm:IC_allw}) as one of two deterrence terms; the
conservative sufficient condition drops it entirely and relies on the
bounty gap alone.

\subsubsection{Bounty share under \textsf{PIVOT-$K$}}
\label{sec:bounty_share}

This subsection uses a stationary Bernoulli thinning benchmark rather than the fluid proxy from Section~\ref{sec:threat}. Conditional on
$A_t$, each cartel-addressed bundle is included independently with
probability $w$. Thus
\[
  X_t^{(w)} \mid A_t \sim \mathrm{Bin}(A_t,w),
  \qquad
  Y_t^{(w)} := H_t + X_t^{(w)} \in \{0,1,\dots,m\}.
\]

This keeps the first $\kappa$ prefix objects integer valued while
preserving the interpretation of $w$ as an average inclusion rate. The
estimates below are used only to obtain a conservative model based
approximation to the cartel's expected bounty share, the exact pathwise
statements used later are Theorem~\ref{thm:worst_w0_rigorous},
Theorem~\ref{thm:dyn_best_binary}, and
Lemma~\ref{lem:pathwise_mono_pivot}.

Define
\[
  \pi(w)
  := \frac{\E[X_t^{(w)}]}{\E[Y_t^{(w)}]}
  = \frac{w\beta}{1-\beta+w\beta}.
\]
Let $J_\kappa^{(w)}$ be the number of cartel bundles among the first
$\kappa$ included bundles under this stationary benchmark.

The expected cartel bounty share concentrates around $\pi(w)B$ with an
additive $O(m/\kappa)$ terminal slot error:
\label{prop:bshare_w}%
\begin{equation}
\label{eq:bshare_bound}
  \left|\E[J_\kappa^{(w)}]-\pi(w)\kappa\right|\le m,
  \qquad\text{hence}\qquad
  \left|
    \E[\text{bounty to }\mathcal{C}\mid w]-\pi(w)B
  \right|
  \le \frac{m}{\kappa}B.
\end{equation}
To see this, set
\[
  M_t := \sum_{u=1}^t \bigl(X_u^{(w)}-\pi(w)Y_u^{(w)}\bigr),
  \qquad
  N := \inf\Bigl\{t:\sum_{u=1}^t Y_u^{(w)}\ge \kappa\Bigr\},
\]
and let
\[
  O := \sum_{u=1}^N Y_u^{(w)}-\kappa \in \{0,\dots,m-1\}.
\]
Because the increments of $M_t$ are mean zero and bounded by $m$,
optional stopping gives
\[
  \E\!\left[\sum_{u=1}^N X_u^{(w)}\right]
  = \pi(w)\,\E\!\left[\sum_{u=1}^N Y_u^{(w)}\right]
  = \pi(w)\bigl(\kappa+\E[O]\bigr).
\]
Only bundles from the terminal slot contribute to the overshoot, so
replacing $\sum_{u=1}^N X_u^{(w)}$ by the prefix count $J_\kappa^{(w)}$
incurs an additive error of at most $m$, yielding
\eqref{eq:bshare_bound}. When $\kappa$ is small, this approximation
should be viewed as a coarse benchmark rather than a sharp estimate.

Under full inclusion ($w=1$), the first $\kappa$ included bundles are
exactly the first $\kappa$ lane contacts in deterministic order. Each
contact hits the cartel with marginal probability $\beta$, so
\label{lem:Jkappa_w1}%
\begin{equation}
\label{eq:Jkappa_w1}
  \E[J_\kappa^{(1)}]=\beta\kappa,
  \qquad\text{and hence}\qquad
  R_{\mathrm{bounty}}(1)=\gamma^{t^\star}\beta B.
\end{equation}

Define the discounted bounty revenue under the stationary benchmark by
\[
  T^{(w)} := \inf\Bigl\{t:\sum_{u=1}^t Y_u^{(w)}\ge \kappa\Bigr\},
  \qquad
  R_{\mathrm{bounty}}(w)
  :=
  \E\!\left[\gamma^{T^{(w)}}\cdot \frac{J_\kappa^{(w)}}{\kappa}\,B\right].
\]
Since $Y_t^{(w)}\le m$ pathwise, no benchmark policy can include earlier than $T(1)=t^\star$, so $\gamma^{T^{(w)}}\le \gamma^{T(1)}$ almost surely. Combining this with \eqref{eq:bshare_bound} gives the
conservative upper bound
\label{prop:bounty_upper}%
\begin{equation}
\label{eq:bounty_upper}
  R_{\mathrm{bounty}}(w)
  \le
  \gamma^{T(1)}\!\left(\pi(w)+\frac{m}{\kappa}\right)B.
\end{equation}
Therefore
\[
  R_{\mathrm{bounty}}(1)-R_{\mathrm{bounty}}(w)
  \ge
  \gamma^{T(1)}\!\left(\beta-\pi(w)-\frac{m}{\kappa}\right)B.
\]

This benchmark bounty gap, together with the fee gap from
Section~\ref{sec:fee_revenue}, is the direct revenue term used in the
stationary IC calculation below. The exact dynamic analysis later relies instead on pathwise monotonicity and the delay benchmark of
Theorem~\ref{thm:worst_w0_rigorous}.

\subsection{Sufficient condition for full inclusion}
\label{sec:IC}

We now combine the payoff bounds into an explicit IC condition: a
constraint on the bounty~$B$ under which the cartel weakly prefers
full inclusion to every stationary partial withholding policy
$w \in [0,1)$.

\begin{theorem}[Sufficient condition for $w{=}1$ to dominate all
  stationary $w{<}1$]
\label{thm:IC_allw}
Fix $(m,s,K)$. Consider the stationary benchmark of Section~\ref{sec:bounty_share}, together with the fluid delay proxy
$q^{(w)}$ from Section~\ref{sec:threat}. A conservative sufficient
condition for full inclusion ($w = 1$) to dominate every stationary
benchmark $w \in [0,1)$ is
\[
  \Bigl(\beta - \pi(w) - \frac{m}{\kappa}\Bigr)\,B
  \;\ge\;
  \alpha v\, q^{(w)}
  \qquad \forall\, w \in [0,1).
\]
\end{theorem}

\begin{proof}
Within this stationary benchmark,
\[
  \Pi_{\mathcal{C}}(w)
  =
  R_{\mathrm{fee}}(w)
  +
  R_{\mathrm{bounty}}(w)
  +
  \alpha v\,G(w),
\]
with $G(w)=\gamma^{t^\star}q^{(w)}$ by definition of the fluid delay
proxy. Full inclusion dominates if
\[
  \bigl(R_{\mathrm{fee}}(1)-R_{\mathrm{fee}}(w)\bigr)
  +
  \bigl(R_{\mathrm{bounty}}(1)-R_{\mathrm{bounty}}(w)\bigr)
  \ge
  \alpha v\,G(w)
\]
for all $w<1$. Dropping the nonnegative fee gap and substituting
\eqref{eq:Jkappa_w1} and \eqref{eq:bounty_upper} gives
\[
  \gamma^{t^\star}
  \Bigl(\beta-\pi(w)-\frac{m}{\kappa}\Bigr)B
  \ge
  \alpha v\,\gamma^{t^\star}q^{(w)},
\]
and dividing by $\gamma^{t^\star}$ yields the claim.
\end{proof}

Theorem~\ref{thm:IC_allw} is a model based per-$w$ condition for the
stationary benchmark, the exact global dynamic reduction remains
Theorem~\ref{thm:worst_w0_rigorous}. Verifying it for every
$w \in [0,1)$ can still be cumbersome.  Two reductions simplify the
right hand side.  First, Theorem~\ref{thm:worst_w0_rigorous} already
establishes $q^\pi \le q^{(0)}$ for all dynamic policies, so one may always replace~$q^{(w)}$ with the full withholding bound~$q^{(0)}$.  Second, on the knife edge $\Delta = 0$, Lemma~\ref{lem:threshold_w} shows that for every $w<1$, $q^{(w)} = \Prb[S_{t^\star} > 0]$, since
$(1-w)S_{t^\star}>0 \iff S_{t^\star}>0$.

These reductions pin down the MEV side (the right hand side) of the IC
inequality, but the direct revenue side still depends on~$w$ through the bounty share function~$\pi(w)$.  In particular, the stationary policy that minimises the left hand side need not be $w = 0$, intermediate~$w$ can reduce the bounty gap while still triggering delay. Rather than optimising over~$w$ in the stationary family, we first establish in
Section~\ref{sec:pivot_optimality} that the uniform weights of
\textsf{PIVOT-$K$} are the strongest choice within the class of
rank based bounty rules, maximising the worst case bounty forfeiture
from any deletion attack. Section~\ref{sec:dynamic_sabotage} then
analyses the cartel's dynamic best response directly, showing that it
reduces to withholding exactly $\Delta + 1$ bundles (minimal sabotage),
and derives the closed form IC condition on the knife
edge~\eqref{eq:dynIC_delta0_B}.

\subsection{Optimality among budget capped rank based rules}
\label{sec:pivot_optimality}

The IC condition of Theorem~\ref{thm:IC_allw} depends on the bounty
gap, which in turn depends on how much bounty revenue the cartel
sacrifices by removing bundles from the pivotal prefix.  A natural
question is whether the uniform per-bundle payment $B/\kappa$ used by
\textsf{PIVOT-$K$} is the best allocation of a fixed budget~$B$ across
the $\kappa$ pivotal positions, or whether concentrating payments on
certain ranks (e.g.\ paying more for earlier bundles) could yield a
stronger deterrent.  This subsection shows that uniform payments are
minimax optimal: no other allocation of the same budget creates a larger worst case bounty forfeiture when bundles are withheld.

Let $b^{(1)},\dots,b^{(\kappa)}$ denote the first $\kappa$ admissible
included bundles in Sedna's deterministic symbol resolution order.
A prefix linear pivotal bundle bounty rule is specified by
nonnegative weights
$\omega = (\omega_1,\dots,\omega_\kappa)$ summing to~$1$, and pays
$\omega_j B$ to the lane proposer of $b^{(j)}$ for each
$j\in\{1,\dots,\kappa\}$, non-pivotal bundles receive nothing.  This
class is deterministic, anonymous (depends only on rank~$j$),
budget capped, and compatible with Sedna's deduplication and lane bound
tickets. \textsf{PIVOT-$K$} corresponds to the uniform weights
$\omega_j\equiv 1/\kappa$.

For any weight vector $\omega$, let
$\omega_{(1)}\le \cdots \le \omega_{(\kappa)}$ be the nondecreasing
rearrangement, and define the $d$-removal floor
\[
  S_d(\omega) \;:=\; \sum_{i=1}^{d} \omega_{(i)},
  \qquad d\in\{1,\dots,\kappa\}.
\]
Operationally, $B\cdot S_d(\omega)$ is the minimum total bounty that
can be eliminated by withholding any $d$~pivotal bundles in the worst
case, because an adversary may target the $d$ least paid ranks.

\begin{theorem}[Minimax optimality of uniform pivotal payments]
\label{thm:pivot_optimal}
Fix $\kappa\ge 1$ and $d\in\{1,\dots,\kappa\}$. Over all prefix linear
pivotal bundle rules $\omega$,
\[
  S_d(\omega) \;\le\; \frac{d}{\kappa},
\]
with equality if and only if $\omega_j\equiv 1/\kappa$ for all $j$.
Equivalently, the uniform rule (\textsf{PIVOT-$K$}) maximizes the
worst case bounty forfeiture caused by withholding any $d$ pivotal
bundles.
\end{theorem}

\begin{proof}
Let $a := S_d(\omega)/d$, the average of the $d$ smallest weights.
Since the weights are sorted nondecreasingly, $\omega_{(i)} \ge a$ for
every $i > d$.  Therefore
\[
  1
  \;=\;
  \sum_{i=1}^{\kappa}\omega_{(i)}
  \;=\;
  S_d(\omega) + \sum_{i=d+1}^{\kappa}\omega_{(i)}
  \;\ge\;
  da + (\kappa-d)a
  \;=\;
  \kappa a,
\]
so $a \le 1/\kappa$ and $S_d(\omega) = da \le d/\kappa$.  If equality
holds then every inequality is tight, forcing all $\kappa$ weights to
equal $1/\kappa$.  Conversely, the uniform rule attains equality.
\end{proof}

In the IC analysis that follows, the relevant deletion count is
$d = \Delta + 1$, the minimal number of withheld bundles needed to
force a delay.  Among all
prefix linear rules, \textsf{PIVOT-$K$} therefore maximises
$S_{\Delta+1}(\omega)$ and hence minimises the bounty budget~$B$
needed to meet any given deterrence threshold.

Theorem~\ref{thm:pivot_optimal} is a minimax statement within the class
of anonymous, budget capped rank based rules.  It does not exclude
richer state dependent mechanisms, but it shows that once the designer
commits to paying a fixed budget across the pivotal prefix, flattening
the weights is the strongest way to maximize the worst case bounty
forfeiture from a deletion attack.

\subsection{Dynamic deviations and minimal sabotage}
\label{sec:dynamic_sabotage}

The stationary $w$-family can overstate the revenue cost of causing
delay.  A dynamic cartel can withhold the minimum number of bundles
needed to push inclusion past~$t^\star$, while including everything else
to preserve fee and bounty revenue.

Recall from Section~\ref{sec:reduction} that under a general cartel
policy~$\pi$, the on-chain inclusion time is
$T^\pi := \inf\{t : U_t^\pi \ge \kappa\}$.  Define the cumulative
number of cartel addressed bundles withheld up to~$t^\star$ as
$W_{t^\star}^\pi := S_{t^\star} - \sum_{t=1}^{t^\star} X_t^\pi$.
By the same accounting used in Lemma~\ref{lem:threshold_w}, the
on-chain count at the honest horizon satisfies
$U_{t^\star}^\pi = (\kappa + \Delta) - W_{t^\star}^\pi$, so delay
occurs if and only if $W_{t^\star}^\pi > \Delta$.  On the knife edge
($\Delta = 0$), a single withheld bundle suffices.

Since each included bundle earns a fixed fee~$f$ and the cartel's
\textsf{PIVOT-$K$} bounty share is pathwise nondecreasing in inclusions
(Lemma~\ref{lem:pathwise_mono_pivot} in
Appendix~\ref{app:dynamic}), withholding beyond the minimum $\Delta + 1$
only reduces direct revenue without increasing the MEV term.

\begin{theorem}[Minimal sabotage is optimal among delay inducing policies]
\label{thm:dyn_best_binary}
Fix a sample path of the lane contact process up to~$t^\star$.
Among all policies that achieve $T^\pi > t^\star$, an optimal
policy withholds exactly $\Delta + 1$ cartel addressed bundles
and includes all others.
\end{theorem}

\begin{proof}
At least $\Delta + 1$ withholdings are needed for delay.  Each additional withheld bundle strictly decreases fee revenue and weakly decreases bounty revenue (Lemma~\ref{lem:pathwise_mono_pivot}), while the MEV term depends only on the binary event $\{T^\pi > t^\star\}$.  Hence exactly $\Delta + 1$ withholdings are optimal.
\end{proof}

Theorem~\ref{thm:dyn_best_binary} relies on the static fee market
assumption (\S\ref{sec:model}) where the net marginal payoff from
including an additional admissible bundle is nonnegative.  In dynamic
fee markets with EIP-1559-style congestion pricing or priority fee
auctions, including a Sedna bundle may displace a higher value
alternative, making the net marginal payoff negative.
Appendix~\ref{app:dynamic} records a concrete counterexample and states
the necessary generalised condition (Assumption~\ref{as:net_nonneg})
that restores minimal sabotage in general fee environments.

On the knife edge ($\Delta = 0$), minimal sabotage reduces to
withholding exactly one bundle whenever the cartel is contacted
($S_{t^\star} > 0$).  Let
\[
  q_0
  \;:=\;
  \Prb[S_{t^\star} > 0]
  \;=\;
  1 - \left(
    \frac{\binom{(1-\beta)n}{m}}{\binom{n}{m}}
  \right)^{\!t^\star}.
\]
Comparing the cartel's payoff under this one bundle deviation against
full inclusion, the direct revenue sacrifice consists of three terms: a
net fee sacrifice from the single withheld bundle (offset by the extra
slot's fee windfall under the static sender), a bounty discount loss
from later inclusion, and the lost pivotal share when a cartel bundle
exits the prefix and an honest replacement enters. Assembling these
and requiring that they exceed the MEV option $\alpha v \cdot
\gamma^{t^\star} q_0$ yields the sufficient condition
\begin{equation}
\label{eq:dynIC_delta0_explicit}
  \underbrace{q_0\bigl(\gamma^{t^\star-1}f
    - \gamma^{t^\star}\beta m f\bigr)
  }_{\text{net fee sacrifice}}
  \;+\;
  \underbrace{\gamma^{t^\star}(1-\gamma)\beta B
  }_{\text{bounty discount loss}}
  \;+\;
  \underbrace{\gamma^{t^\star+1}(1-\beta)\frac{q_0}{\kappa}B
  }_{\text{lost pivotal share}}
  \;\ge\;
  \alpha v \cdot \gamma^{t^\star} q_0,
\end{equation}
which rearranges to the closed-form bounty threshold
\begin{equation}
\label{eq:dynIC_delta0_B}
  B
  \;\ge\;
  \frac{q_0\bigl(\alpha v - f/\gamma + \beta m f\bigr)}
       {(1-\gamma)\beta
        + \gamma(1-\beta)\,q_0/\kappa}\,.
\end{equation}
The fee term assumes the static sender continues disseminating in
slot~$t^\star + 1$; if the sender adapts, the cartel's extra-slot
windfall $\gamma^{t^\star}\beta m f$ should be reduced.  Dropping the
pivotal share term (third brace) yields a more conservative sufficient
condition that does not depend on the replacement bundle's lane
assignment.

Equations~\eqref{eq:dynIC_delta0_explicit}--\eqref{eq:dynIC_delta0_B}
are the binding IC constraint for knife edge instances ($\Delta = 0$),
where the attack is unilateral and no coalition decomposition is needed. For $\Delta > 0$, the situation is sharply different as the attack requires coordinating at least $\Delta+1$ withheld bundle
opportunities and, under the static sender, these need not correspond
to $\Delta+1$ distinct lanes. The next subsection shows that the total
extractable MEV from this transaction cannot cover the resulting
aggregate opportunity cost.

\subsection{Local rationality and coalition stability}
\label{sec:coalition}

The preceding analysis treats the cartel as a monolithic actor.  We now
decompose the incentive problem to the level of individual withholding
opportunities. Fees and bounties are lane bound, but under the static
sender the same lane may be contacted again in later slots, so the
statements below are phrased in terms of withheld bundles rather than
distinct cartel lanes.

When $\Delta > 0$, a single withheld bundle cannot cause any delay.  If
exactly one received bundle is withheld while all others are included
honestly, then $W_{t^\star} = 1 \le \Delta$, so no delay occurs and
the cartel obtains no MEV option from this deviation.  The deviator
therefore sacrifices the fee~$f$ and, if the withheld bundle would
have landed inside the pivotal prefix, the bounty share $B/\kappa$,
in exchange for no offsetting gain.  Any one-bundle deviation is
therefore strictly unprofitable whenever $f > 0$, and weakly dominated
in general unless the deviator is subsidised externally.  This
observation is the core of the coalition decomposition: positive slack
turns the delay attack from a unilateral option into a coordination
problem, because the first $\Delta$ withholdings each destroy revenue
without producing any MEV to share.
% When $\Delta > 0$, a single withheld bundle cannot cause any delay.  If
% exactly one received bundle is withheld while all others are included
% honestly, then $W_{t^\star} = 1 \le \Delta$, so no delay occurs.  The deviator gains zero MEV and sacrifices the fee~$f$ (plus a bounty share if the withheld bundle
% would have landed inside the pivotal prefix). Any one bundle deviation
% is therefore strictly unprofitable whenever $f > 0$, and weakly
% dominated in general unless the deviator is subsidised externally.
% This observation is the core of the coalition decomposition: positive
% slack turns the delay attack from a unilateral option into a
% coordination problem.

\begin{lemma}[Coalition budget constraint]
\label{prop:coalition_IC}
Fix $\Delta>0$ and let a coalition withhold $c\ge \Delta+1$ bundles
before $t^\star$.  Among the $\kappa+\Delta$ bundle positions present
at the honest horizon, at most $\Delta$ are redundant (outside the
pivotal prefix).  Deleting $c$ bundles therefore removes at least
$(c-\Delta)^+$ positions from the pivotal prefix.  The coalition's
total direct revenue sacrifice is consequently at least
\begin{equation}
\label{eq:coalition_loss}
  c\,f + \frac{(c-\Delta)^+}{\kappa}B.
\end{equation}
If the bounty satisfies
\begin{equation}
\label{eq:coalition_budget}
  (\Delta+1)f + \frac{B}{\kappa}
  \;>\;
  \alpha v \cdot \gamma^{t^\star},
\end{equation}
then no delay inducing coalition can cover its aggregate opportunity
cost from the MEV of this single transaction.
\end{lemma}

\begin{proof}
Any coalition that delays the transaction must withhold
$c \ge \Delta+1$ bundles, so its
direct loss is at least $(\Delta+1)f + B/\kappa$.  If this exceeds
$\alpha v \cdot \gamma^{t^\star}$, the total redistributable MEV is
smaller than the coalition's aggregate sacrifice, and the deviation
cannot be jointly profitable.
\end{proof}

Rearranging~\eqref{eq:coalition_budget}, a sufficient bounty is
\begin{equation}
\label{eq:coalition_B}
  B \;\ge\; \kappa\!\left(
    \alpha v \cdot \gamma^{t^\star} - (\Delta+1)f
  \right)^{\!+}.
\end{equation}

The budget constraint~\eqref{eq:coalition_budget} holds under any
internal sharing rule (equal splitting, proportional, or arbitrary side
payments) as long as the total redistributable value from this
transaction is at most $\alpha v \cdot \gamma^{t^\star}$.  MEV
extraction is performed by a single proposer who inserts a front-running
transaction, and the total value it can redistribute is capped by what
it extracts.  An off-chain MEV relay or builder market that pools
extraction gains across multiple transactions could in principle
increase the effective pie beyond $\alpha v \cdot \gamma^{t^\star}$,
restoring the monolithic threat.  In a permissionless setting without
such cross transaction subsidy infrastructure, the aggregate budget
bound applies and the monolithic IC overstates the threat for all
$\Delta > 0$ instances.

For $\Delta = 0$, a single lane can unilaterally cause delay and extract
MEV without coordination, the coalition decomposition does not apply and
the monolithic IC
(equations~\eqref{eq:dynIC_delta0_explicit}--\eqref{eq:dynIC_delta0_B})
remains the binding constraint.

\section{Adaptive sender strategies}
\label{sec:ratchet}

Sections~\ref{sec:threat}--\ref{sec:incentives} assumed a static sender
that blindly contacts $m$~uniformly sampled lanes every slot.  In
practice, the sender observes on-chain inclusion after each slot and can
exclude non-responsive lanes.  This makes withholding self-revealing where every withheld bundle burns the cartel lane's cover.  We formalise this ratchet and show it strictly reduces the multi-slot delay probability when $t^\star \ge 2$.

Under lane bound tickets, the sender knows exactly which lanes were
contacted in slot~$t$ and can verify after finalisation whether each
issued ticket was redeemed in an admissible bundle occurrence.  Honest
contacted lanes always include (they earn fee revenue~$f$ and have no
MEV motive to withhold), so non-redemption is a perfect detection
signal for cartel behaviour.

The sender maintains an eligible lane set
$\mathcal{E}_t \subseteq \{1,\dots,n\}$ for this transaction.
Initially $\mathcal{E}_1 = \{1,\dots,n\}$.  In slot~$t$ the sender
samples $m_t$ lanes uniformly from~$\mathcal{E}_t$ and sends a bundle
plus ticket to each.  After finalisation, any contacted lane whose
ticket was not redeemed is removed:
\[
  \mathcal{E}_{t+1}
  := \mathcal{E}_t
  \setminus
  \bigl\{\ell \in C_t :
    \text{no admissible redemption of the ticket for $(\id, t, \ell)$}
  \bigr\},
\]
where $C_t$ is the set of lanes contacted in slot~$t$.

In practice, honest lanes may occasionally fail to redeem due to
network jitter or downtime.  Appendix~\ref{sec:adaptive_sender} models
this with a per-contact miss rate~$\varepsilon$ and shows that the
ratchet's delay bound degrades gracefully: the delay probability is
bounded by
$\Prb[W + \mathrm{Bin}(M_{t^\star}, \varepsilon) >
\Delta_{\mathrm{rec}}]$, reducing to the perfect detection case when
$\varepsilon = 0$.

\subsection{First slot deficit bound}
\label{sec:firstslot}

Consider a one shot deviation where the cartel withholds only in
slot~$1$ and all withholding lanes are flagged and excluded.  In
slots $t = 2, \dots, t^\star$ the sender contacts $m_t$~lanes from the
reduced pool~$\mathcal{E}_t$, with total planned contacts and recovery
slack
\[
  M_{t^\star} := m + \sum_{t=2}^{t^\star} m_t,
  \qquad
  \Delta_{\mathrm{rec}} := M_{t^\star} - \kappa.
\]
Since only slot~$1$ suffers withholding, the on-chain count by
time~$t^\star$ is $M_{t^\star} - W_1$, where $W_1$ is the number of
cartel bundles withheld in slot~$1$.  Inclusion by~$t^\star$ therefore
holds iff $W_1 \le \Delta_{\mathrm{rec}}$. Under maximal first slot
withholding ($W_1 = A_1$), the ratchet corrected delay probability is
the hypergeometric tail
\label{prop:firstslot_deficit}%
\begin{equation}
\label{eq:q_rat}
  q_{\mathrm{rat}}
  \;:=\;
  \Prb[A_1 > \Delta_{\mathrm{rec}}],
  \qquad
  A_1 \sim \mathrm{Hypergeom}(n,\, \beta n,\, m).
\end{equation}
Writing $a := \Delta_{\mathrm{rec}} + 1$ and $\theta := a/m$, the tail
bounds~\eqref{eq:kl_tails} with $t = 1$ give
$q_{\mathrm{rat}} \le \exp\{-m\,\KL(\theta \| \beta)\}$ whenever
$\theta > \beta$.

Note that in slot~$1$, $A_1 + H_1 = m$, so single slot liveness
(enough honest inclusions) and single slot privacy risk (enough cartel
contacts) are complementary tails of the same hypergeometric draw.  For
a sender targeting $t^\star = 1$ with $\kappa \le m$, reducing liveness
failure $\Prb[A_1 > m - \kappa]$ automatically reduces the one shot
delay probability.

\subsection{Adversary's constrained best response}
\label{sec:adv_ratchet}

The one shot bound of Section~\ref{sec:firstslot} treats the simplest
deviation. A strategic adversary can spread withholding across multiple
slots, sacrificing different lanes each slot to maintain partial cover.
The ratchet constrains this by making every sacrifice permanent.

Let $W_t$ denote the number of cartel lane contacts withheld (and
consequently flagged) in slot~$t$, and let $F_t := \sum_{u=1}^t W_u$ be the cumulative number of cartel lanes removed from the eligible set by the end of slot~$t$.  After slot~$t$, the eligible pool has $n_t := n - F_t$ lanes of which at most $\beta n - F_t$ are cartel, giving an effective cartel fraction
\label{lem:beta_shrink}%
\begin{equation}
\label{eq:beta_shrink}
  \beta_t
  \;:=\;
  \frac{\beta n - F_t}{n - F_t}
  \;=\;
  \beta - \frac{F_t(1-\beta)}{n - F_t}
  \;<\; \beta
  \qquad\text{for } F_t \ge 1.
\end{equation}
Since the subtracted term is strictly increasing in~$F_t$, the
effective cartel fraction shrinks with every flagged lane.

Creating delay requires withholding at least $\Delta + 1$
cartel addressed bundles across all slots up to~$t^\star$.  Under the
ratchet, each withheld bundle costs a permanently excluded cartel lane.
The adversary's total attack budget across all transactions sharing the
same lane pool is therefore at most $\beta n$ lane sacrifices.

\begin{lemma}[Ratchet corrected delay probability for
  $t^\star \ge 2$]
\label{prop:ratchet_delay}
Fix $t^\star \ge 2$ and suppose the adversary spreads its $\Delta + 1$
required withholdings across slots $1, \dots, t^\star$.  Let
$A_t^{(\mathrm{rat})}$ be the number of cartel contacts in slot~$t$
under the reduced pool $(n_t, \beta_t)$ defined
by~\eqref{eq:beta_shrink}. The delay probability satisfies
\[
  q_{\mathrm{rat}}^{(\mathrm{multi})}
  \;\le\;
  \Prb\!\left[
    \sum_{t=1}^{t^\star} A_t^{(\mathrm{rat})} > \Delta
  \right],
\]
with $n_t \le n$ and $\beta_t \le \beta$ for all $t \ge 2$.  In
particular, $q_{\mathrm{rat}}^{(\mathrm{multi})} < q^{(0)}$ for any
strategy that withholds at least one bundle before slot~$t^\star$.
\end{lemma}

\begin{proof}
Under the ratchet, the pool shrinks after each withholding slot:
$\beta_t < \beta$ for $t \ge 2$ by~\eqref{eq:beta_shrink}.  Since the
hypergeometric upper tail is stochastically increasing in the cartel
fraction, each $A_t^{(\mathrm{rat})}$ is stochastically dominated by
$A_t \sim \mathrm{Hypergeom}(n, \beta n, m)$, with strict domination
for $t \ge 2$ whenever $F_{t-1} \ge 1$.  Summing gives the stated
bound; the strict inequality follows because at least one slot has a
strictly smaller effective~$\beta$.
\end{proof}

Under the ratchet, each withheld bundle permanently flags the
corresponding lane.  For a coalition that withholds
$c \ge \Delta + 1$ bundles, the direct loss lower
bound~\eqref{eq:coalition_loss} captures only the within transaction
fee and bounty sacrifice.  Ratcheting adds a further dynamic cost as
every flagged lane forgoes future fee revenue and, when relevant,
future bounty revenue on subsequent transactions.  This strictly
strengthens the aggregate budget constraint~\eqref{eq:coalition_budget} while leaving the single
transaction MEV pie unchanged.

The preceding analysis grants the adversary full knowledge of the
decode threshold~$\kappa$ (and hence the slack~$\Delta$).  Under Sedna,
$\kappa = \lceil K/s \rceil$ depends on the payload size, which is
hidden behind a computationally hiding commitment until decode.  For
$t^\star = 1$ this is immaterial: after slot-$1$ finalisation the
adversary can pool its $A_1$~private bundles with the $H_1$~on-chain
honest bundles, obtaining all $m \ge \kappa$ symbols and decoding
regardless of~$\kappa$.  For $t^\star \ge 2$, however, the ratchet
forces the adversary to commit its slot-$1$ withholding before any
on-chain evidence of~$\kappa$ accumulates: the adversary does not know
how many bundles are needed for reconstruction, and every withheld
bundle permanently shrinks the cartel's effective
fraction~\eqref{eq:beta_shrink}.  By the time the adversary could
narrow its belief about~$\kappa$ from observing continued dissemination
in slot~$2$, its slot-$1$ lanes are already excluded.  The maximal
withholding assumption ($W_1 = A_1$) in Section~\ref{sec:firstslot} is
therefore conservative: a rational adversary uncertain about~$\Delta$
may withhold strictly fewer bundles, since each unnecessary sacrifice
permanently burns a cartel lane for no MEV gain whenever~$\Delta$ turns
out to be large.  Quantifying the optimal Bayesian withholding strategy
under a prior over~$\kappa$ is a natural extension that could yield
tighter IC bounds.

Since every sufficient condition in Section~\ref{sec:incentives} is
monotone in the delay probability, replacing $q^{(w)}$ with the
ratchet corrected $q_{\mathrm{rat}}$ from~\eqref{eq:q_rat} or
Lemma~\ref{prop:ratchet_delay} yields a valid IC condition with a
strictly lower bounty threshold whenever $t^\star \ge 2$.
\section{Intra-slot decode races}
\label{sec:micro}

Under full inclusion, the coarse timing model
(Assumption~\ref{as:coarse}) eliminates multi-slot information leads
($q^{(1)} = 0$; \S\ref{sec:partial}).  The residual risk is a
within slot race where in the final dissemination slot~$t^\star$, the cartel may receive enough symbols to decode the payload before the slot is sealed and finalised.  This section bounds the probability that such a race succeeds, first from above (showing it is exponentially small in~$m$ when $\Delta > 0$) and then from below (showing that no incentive only mechanism can eliminate it).

Recall $t^\star = \lceil \kappa/m \rceil$ and
$\Delta = t^\star m - \kappa$.  Define
\[
  r \;:=\; m - \Delta \;=\; \kappa - (t^\star - 1)m,
\]
the number of bundles the cartel still needs at the start of
slot~$t^\star$.  Under full inclusion, exactly $(t^\star - 1)m$
admissible bundles have been finalised by the end of
slot~$t^\star - 1$, so~$r$ is deterministic.  Since at most $tm$
bundles exist after~$t$ slots, decoding is impossible before
slot~$t^\star$; the within slot race can only occur in the final slot.

During slot~$t^\star$ the sender disseminates $m$~bundles, of which
$A_{t^\star}$ are addressed to cartel lanes.  A within slot private
decode before sealing requires $A_{t^\star} \ge r$.  We abstract the
race dynamics into a \emph{conditional success function}
$\rho(a, r) \in [0, 1]$, representing the probability the cartel
decodes and acts before the slot is sealed, given that $a$~bundles
are addressed to cartel lanes and $r$~bundles are still needed.  Let
\[
  \bar\rho \;:=\; \sup_{a \ge r,\, r \ge 1} \rho(a, r)
\]
denote the worst case conditional success probability, taken over all
feasible values of cartel contacts and deficit.
Appendix~\ref{app:micro_pipeline} instantiates~$\rho$ via a concrete
sealing deadline pipeline that matches any slot based BFT design.

Define the discounted within slot MEV term under full inclusion as
\[
  G_{\mathrm{inc}} :=
  \gamma^{t^\star - 1}\,
  \Prb[\text{within slot MEV succeeds in slot } t^\star].
\]
Since success requires both feasibility ($A_{t^\star} \ge r$) and
conditional success (at most~$\bar\rho$), we obtain the following
bound.

\begin{lemma}[Within slot MEV upper bound]
\label{lem:within_slot_upper}
Under full inclusion,
\begin{equation}
\label{eq:Ginc_upper}
  G_{\mathrm{inc}}
  \;\le\;
  \bar\rho \cdot \gamma^{t^\star - 1}
  \cdot \Prb[A_{t^\star} \ge r].
\end{equation}
The hypergeometric tail is bounded by the KL
exponents~\eqref{eq:kl_tails} with $t = 1$: whenever $r/m > \beta$,
\[
  \Prb[A_{t^\star} \ge r]
  \;\le\;
  \exp\!\bigl\{-m\,\KL(r/m \| \beta)\bigr\}.
\]
On the knife edge ($\Delta = 0$, $r = m$) the bound tightens to the
exact probability
$\Prb[A_{t^\star} = m] =
\binom{\beta n}{m}\!/\binom{n}{m} \le \beta^m$.  For moderate~$m$ and
$\beta < 1/3$, within slot feasibility is therefore exponentially
unlikely.
\end{lemma}

\begin{proof}
On any sample path, the cartel can attempt a within slot decode only if
$A_{t^\star} \ge r$, and succeeds with probability at most~$\bar\rho$
conditional on feasibility.  The tail bound follows
from~\eqref{eq:kl_tails}.
\end{proof}

Lemma~\ref{lem:within_slot_upper} assumes a specific visibility
model ($V = A_{t^\star}$).  More generally, let $V$ be the total
number of useful bundles the cartel obtains strictly before sealing in
slot~$t^\star$; if the cartel has partial preseal visibility into
honest lanes via gossip, $V$ may exceed~$A_{t^\star}$.  Define the
feasible race floor
$\underline{\rho}(r) := \inf_{v \ge r} \rho(v, r)$.

\begin{theorem}[Within slot MEV floor under incentive only mechanisms]
\label{thm:within_slot_lower}
Fix full inclusion and consider any mechanism whose transfers are
measurable with respect to finalised on-chain history. Then
\[
  G_{\mathrm{inc}}
  \;\ge\;
  \gamma^{t^\star - 1}\,
  \underline{\rho}(r)\,\Prb[V \ge r].
\]
\end{theorem}

\begin{proof}
On the event $\{V \ge r\}$, the cartel has enough preseal information
to attempt decoding, and succeeds with probability at least
$\underline{\rho}(r)$.  A finalization based mechanism cannot
retroactively alter whether $V \ge r$ occurred, since that event is
determined by network propagation, sealing deadlines, and computation
latency before finalisation.
\end{proof}

Together, Lemma~\ref{lem:within_slot_upper} and Theorem~\ref{thm:within_slot_lower} delineate the limits of mechanism
design for the within slot race.  For $\Delta > 0$, within slot MEV
feasibility decays exponentially in~$m$.  However,
Theorem~\ref{thm:within_slot_lower} demonstrates that no
finalization based incentive mechanism can drive this residual to zero.
The pre-seal information event $\{V \ge r\}$ is driven purely by
network propagation and latency, not economics.

This impossibility is restricted to the incentive layer.
Protocol layer mechanisms that fundamentally alter the visibility~$V$
or the success function~$\rho$ can sidestep the bound.  For instance,
commit reveal schemes forcing a cartel to commit before learning
preseal symbols, or threshold encryption demanding a post finalization
key reconstruction round, actively deny the cartel actionable
information.  The lower bound thus marks a clean boundary and everything below this feasibility floor requires cryptographic or timing mechanisms rather than economic ones.

The practical implication is a separation of responsibilities.
Sedna's incentive layer (\textsf{PIVOT-$K$}, the adaptive sender, and
the coalition decomposition) addresses the multi-slot delay threat
analysed in Sections~\ref{sec:incentives}--\ref{sec:ratchet}.  The
within slot residual identified here requires protocol layer
mitigation with tight sealing deadlines, commit reveal schemes, or
threshold encryption. The two layers are complementary and, for
$\Delta > 0$, both risks decay exponentially in~$m$.

\section{Equilibrium viability}
\label{sec:viability}

Section~\ref{sec:incentives} showed \textsf{PIVOT-$K$} achieves IC for
the right parameters, Section~\ref{sec:ratchet} showed the adaptive
sender reduces the required bounty; and Section~\ref{sec:micro}
established that the residual within slot risk is exponentially small
but irreducible by incentives alone.  We now ask whether the three
equilibrium parameters that sustain full inclusion, the proposer fee
share~$\phi$, the decode threshold in bundles~$\kappa$, and the
sender posted bounty~$B$, can actually be set feasibly.

\paragraph{When fees alone suffice.}
Recall that $\phi \in [0,1]$ is the fraction of the per-bundle
bandwidth fee that accrues to the lane proposer, so
$f = \phi\,c\,L(s)$.  Many protocol designs burn a fraction of
transaction fees, leaving only this share for the proposer. Using the fee upper bound from Section~\ref{sec:fee_revenue}, a conservative sufficient condition, within the stationary benchmark of Section~\ref{sec:bounty_share}, for fees alone ($B = 0$) to deter full
withholding is
\label{prop:phi_threshold}%
\begin{equation}
\label{eq:phi_threshold}
  \phi
  \;\ge\;
  \phi^\star
  \;:=\;
  \frac{\alpha v\,\gamma^{t^\star}\,q^{(0)}\,(1 - \gamma)}
       {c\,L(s)\,\beta m\,(1 - \gamma^{t^\star})}.
\end{equation}
As $\phi \to 0$, bounties must carry the full incentive load.  Three
features make \textsf{PIVOT-$K$} bounties a more efficient deterrent
than fees.  First, fee revenue accrues on every included bundle,
including the $\Delta$~redundant ones that do not contribute to
decoding, so the sacrifice from withholding is spread thinly; the
bounty is concentrated on the $\kappa$~pivotal bundles, creating a
larger per-bundle opportunity cost precisely where it counts. Second,
fee tickets are slot-bound when a ticket not redeemed within its validity window expires, so the fee forfeited by withholding is a sunk loss that the proposer cannot recover in a later slot.  The bounty, by
contrast, is paid upon inclusion to whoever fills the pivotal slot; if
one lane withholds, the next admissible bundle from a different lane
claims the reward.  The bounty therefore creates a replacement threat
that fees lack: withholding does not merely forfeit revenue, it
transfers the reward to a competitor.  Third, for protocols with
aggressive fee burning ($\phi \ll 1$), the fee deterrent vanishes while
the bounty is unaffected by the burn rate, making bounties the natural
complement.

\paragraph{Diminishing returns from increasing $\kappa$.}
The sender chooses $\kappa$ (via the symbol count~$s$ and threshold~$K$)
to balance on-chain bandwidth cost against delay risk.  Under a static
sender, increasing $\kappa$ does not suppress the full withholding delay probability, it drives it toward~$1$.  The slack
$\Delta = t^\star m - \kappa$ is bounded by $m - 1$ regardless
of~$\kappa$, while the cumulative cartel contacts $S_{t^\star}$ grow
linearly in~$t^\star$.  The KL tail bounds~\eqref{eq:kl_tails} applied
to the lower tail show that the probability of the cartel contacts
$S_{t^\star}$ remaining at or below the slack~$\Delta$ decays
exponentially as the honest horizon grows:
\label{thm:logkappa_return}%
\begin{equation}
\label{eq:no_delay_upper}
  1 - q^{(0)}(\kappa)
  \;=\;
  \Prb[S_{t^\star} \le \Delta]
  \;\le\;
  \exp\!\left\{
    -t^\star m\,\KL\!\left(
      \frac{\Delta}{t^\star m}\middle\|\beta
    \right)
  \right\}.
\end{equation}
As $\kappa \to \infty$, $t^\star \to \infty$ while
$\Delta/(t^\star m) \to 0$, so the right hand side vanishes and
$q^{(0)}(\kappa) \to 1$. Increasing~$\kappa$ is self-defeating since
bandwidth cost grows linearly while the required bounty plateaus near
$\alpha v / \beta$.  Under the adaptive sender
(Section~\ref{sec:ratchet}), the effective cartel fraction shrinks
after each withholding slot, breaking the pessimistic limit.  The
ratchet makes moderate~$\kappa$ with $\Delta > 0$ the efficient
operating point.

\paragraph{Sender individual rationality.}
The bounty is paid by the transaction sender. Posting a bounty is
rational only if the resulting improvement in inclusion speed and MEV
protection outweighs the bounty cost.  We make this precise.  Under a
stationary cartel policy~$w$ and bounty~$B$, the sender's utility (net
of bandwidth fees) is
\[
  U(w, B)
  \;:=\;
  \E[\gamma^{T(w)}]\,v
  \;-\; B
  \;-\; \alpha v\,G(w).
\]
Let $w^\dagger$ be the cartel's best response at $B = 0$.  Posting a
bounty $B^\star$ that induces full inclusion is individually rational
iff the bounty does not exceed the sender's gain from eliminating delay
and MEV loss:
\label{prop:IR}%
\begin{equation}
\label{eq:IR}
  B^\star
  \;\le\;
  \underbrace{v\bigl(\gamma^{t^\star}
    - \E[\gamma^{T(w^\dagger)}]\bigr)}_{\text{delay cost eliminated}}
  \;+\;
  \underbrace{\alpha v\,G(w^\dagger)}_{\text{MEV loss eliminated}}.
\end{equation}
Higher MEV exposure (larger $\alpha v$) strictly increases willingness
to pay, so \textsf{PIVOT-$K$} naturally accommodates heterogeneous
senders: those whose payloads carry greater MEV risk rationally post
larger bounties.

Computing $w^\dagger$ exactly requires solving the cartel's dynamic
programme. A conservative sufficient condition replaces $w^\dagger$
with $w = 0$:
\begin{equation}
\label{eq:IR_worst}
  B^\star
  \;\le\;
  v\bigl(\gamma^{t^\star} - \E[\gamma^{T(0)}]\bigr)
  \;+\;
  \alpha v\,\gamma^{t^\star}\,q^{(0)}.
\end{equation}
Under the ratchet, $q_{\mathrm{rat}} < q^{(0)}$ lowers the required
$B^\star$, expanding the set of senders for whom posting a bounty is
individually rational.

\section{Evaluation}
\label{sec:numerics}

We instantiate the theory with $n = 100$ lanes, $m = 20$ contacts per
slot, cartel fraction $\beta = 0.2$, and normalised MEV exposure
$\alpha v / \beta = 500$ (so $\alpha v = 100$).  Throughout the
numerics we work in normalized fee units and then translate into USD
for concreteness.  We set $f = 1$ as the unit of proposer revenue per
included bundle and, for the last column of the tables, map this to
\$0.10.  Under the standing normalization $\alpha v/\beta = 500$ with
$\beta = 0.2$, the illustrative transaction has $\alpha v = 100$.
This is only a scale choice, every bounty number reported below scales
linearly with both $f$ and $\alpha v$.
\begin{table}[ht]
\centering
\caption{Derived quantities used in the numerical illustrations.}
\label{tab:numerics_notation}
\small
\begin{tabular}{@{}p{0.28\linewidth}p{0.66\linewidth}@{}}
\toprule
\textbf{Symbol} & \textbf{Meaning} \\
\midrule
$q_{\mathrm{micro}}$ & within-slot feasibility tail 
  $\Prb[A_{t^\star} \ge r]$ (Lemma~\ref{lem:within_slot_upper} 
  with $\bar\rho = 1$) \\
$B^\star_{\mathrm{static}}$ & conservative monolithic bounty proxy 
  $(\alpha v / \beta)\,q^{(0)}$, ignoring fees and discounting \\
$B^\star_{\mathrm{ratchet}}$ & ratchet-corrected bounty proxy 
  $(\alpha v / \beta)\,q_{\mathrm{rat}}$ \\
$B_{\mathrm{coal}}$ & coalition sufficient bounty 
  $\kappa\,\max(0,\;\alpha v\,\gamma^{t^\star} - (\Delta+1)f)$ 
  from~\eqref{eq:coalition_B} \\
\bottomrule
\end{tabular}
\end{table}
Throughout this section, we grant the cartel full knowledge of all
sender parameters ($m$, $s$, $\kappa$, and hence $\Delta$).  As noted
in Section~\ref{sec:instruments}, this is conservative: in practice
$\kappa$ is hidden until decode, and the cartel must commit its
withholding before learning the slack.  The numerical values therefore
represent worst case bounds.

Table~\ref{tab:main} reports exact hypergeometric values for selected
$\kappa$.  For the within slot column we report the feasibility tail
$q_{\mathrm{micro}} := \Prb[A_{t^\star}\ge r]$, which is the upper
bound from Lemma~\ref{lem:within_slot_upper} under the
worst-case assumption that a feasible race always succeeds
($\bar\rho = 1$, i.e.\ any cartel that obtains enough preseal bundles
decodes and acts before sealing with certainty).  The conservative
static bounty proxy is
\[
  B^\star_{\mathrm{static}}
  \;\approx\;
  \frac{\alpha v}{\beta}\,q^{(0)}
  \;=\;
  500\,q^{(0)},
\]
ignoring fee deterrence and discounting.

\begin{table}[ht]
\centering
\caption{Exact delay, ratchet, and within slot feasibility
  probabilities ($n{=}100$, $m{=}20$, $\beta{=}0.2$,
  $\alpha v{=}100$, unit $f{=}1 \leftrightarrow \$0.10$).
  $q_{\mathrm{micro}}$ is the within slot feasibility tail
  $\Prb[A_{t^\star}\ge r]$ under the pessimistic normalization
  $\bar\rho = 1$.  $B^\star_{\mathrm{static}}$ is the conservative
  monolithic bounty proxy $(\alpha v / \beta)\,q^{(0)}$.}
\label{tab:main}
\begin{tabular}{@{}rrrccccr@{}}
\toprule
$\kappa$ & $t^\star$ & $\Delta$ &
  $q^{(0)}$ &
  $q_{\mathrm{rat}}$ &
  $q_{\mathrm{micro}}$ &
  $B^\star_{\mathrm{static}}$ &
  $B^\star_{\mathrm{static}}$ (USD) \\
\midrule
10  & 1 & 10 & $8.0\!\times\!10^{-5}$ & $8.0\!\times\!10^{-5}$ &
  $6.5\!\times\!10^{-4}$ & $0.04$ & ${\sim}\$0$ \\
20  & 1 & 0  & $0.993$ & $0.993$ &
  $1.9\!\times\!10^{-21}$ & $497$ & $\$49.70$ \\
30  & 2 & 10 & $0.136$ & $8.0\!\times\!10^{-5}$ &
  $6.5\!\times\!10^{-4}$ & $68$ & $\$6.80$ \\
50  & 3 & 10 & $0.699$ & $8.0\!\times\!10^{-5}$ &
  $6.5\!\times\!10^{-4}$ & $350$ & $\$35.00$ \\
100 & 5 & 0  & $\approx 1$ & $0.993$ &
  $1.9\!\times\!10^{-21}$ & $500$ & $\$50.00$ \\
\bottomrule
\end{tabular}
\end{table}

For $t^\star = 1$, the ratchet provides no within transaction benefit
and both delay columns coincide.  For $t^\star \ge 2$ with
$\Delta > 0$, the ratchet replaces a $t^\star$-slot accumulation with a
single slot tail, $q_{\mathrm{rat}} = \Prb[A_1 > \Delta]$.  The
improvement is dramatic (e.g.\ $0.136 \to 8 \times 10^{-5}$ for
$\kappa = 30$). The knife edge row $\kappa = 100$ still shows
$q_{\mathrm{rat}} \approx 0.993$ because $\Delta = 0$ means a single
cartel contact in slot~$1$ suffices. However, unlike the static case,
the ratchet prevents this deficit from compounding into later slots,
every withheld lane is permanently excluded, shrinking the cartel's
effective fraction for the remainder of the transaction.  This
distinction matters only for $t^\star \ge 2$, and even then only if
$\Delta > 0$ in the residual slots; on a pure knife edge with
$\Delta = 0$ in every slot, the ratchet cannot recover.

Figure~\ref{fig:static_vs_ratchet} plots $q^{(0)}$ and
$q_{\mathrm{rat}}$ across $\kappa \in [1, 120]$.  The sawtooth pattern
of $q^{(0)}$ is the central operational insight: delay risk climbs
toward~$1$ as $\kappa$ increases within each period of~$m$ (slack
$\Delta$ shrinks), then drops sharply when $\kappa$ crosses a multiple
of~$m$ and a new slot opens (slack resets to $m{-}1$).  The ratchet
flattens the envelope for $t^\star \ge 2$, staying at the single slot
tail level except on knife edges.
Figure~\ref{fig:delay_vs_micro} shows the complementary risk where
$q_{\mathrm{micro}}$ rises exactly where $q^{(0)}$ falls, confirming
that no single $\kappa$ choice eliminates both multi-slot delay and
within slot race risk simultaneously.

\begin{figure}[ht]
\centering
\includegraphics[width=0.85\textwidth]{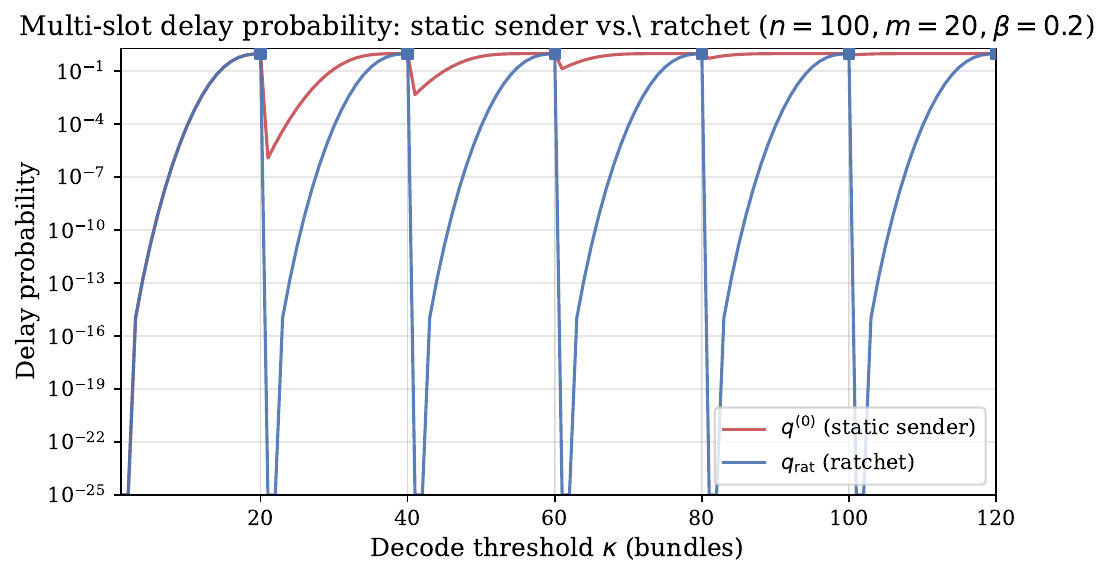}
\caption{Multi-slot delay probability under full withholding, static
  sender ($q^{(0)}$, red) vs.\ ratchet ($q_{\mathrm{rat}}$, blue).
  Dots mark knife edge instances ($m \mid \kappa$).  For
  $t^\star \ge 2$, the ratchet reduces the delay probability by many
  orders of magnitude except on knife edges ($\Delta = 0$) where a
  single cartel contact suffices.
  Parameters: $n = 100$, $m = 20$, $\beta = 0.2$.}
\label{fig:static_vs_ratchet}
\end{figure}

\begin{figure}[ht]
\centering
\includegraphics[width=0.85\textwidth]{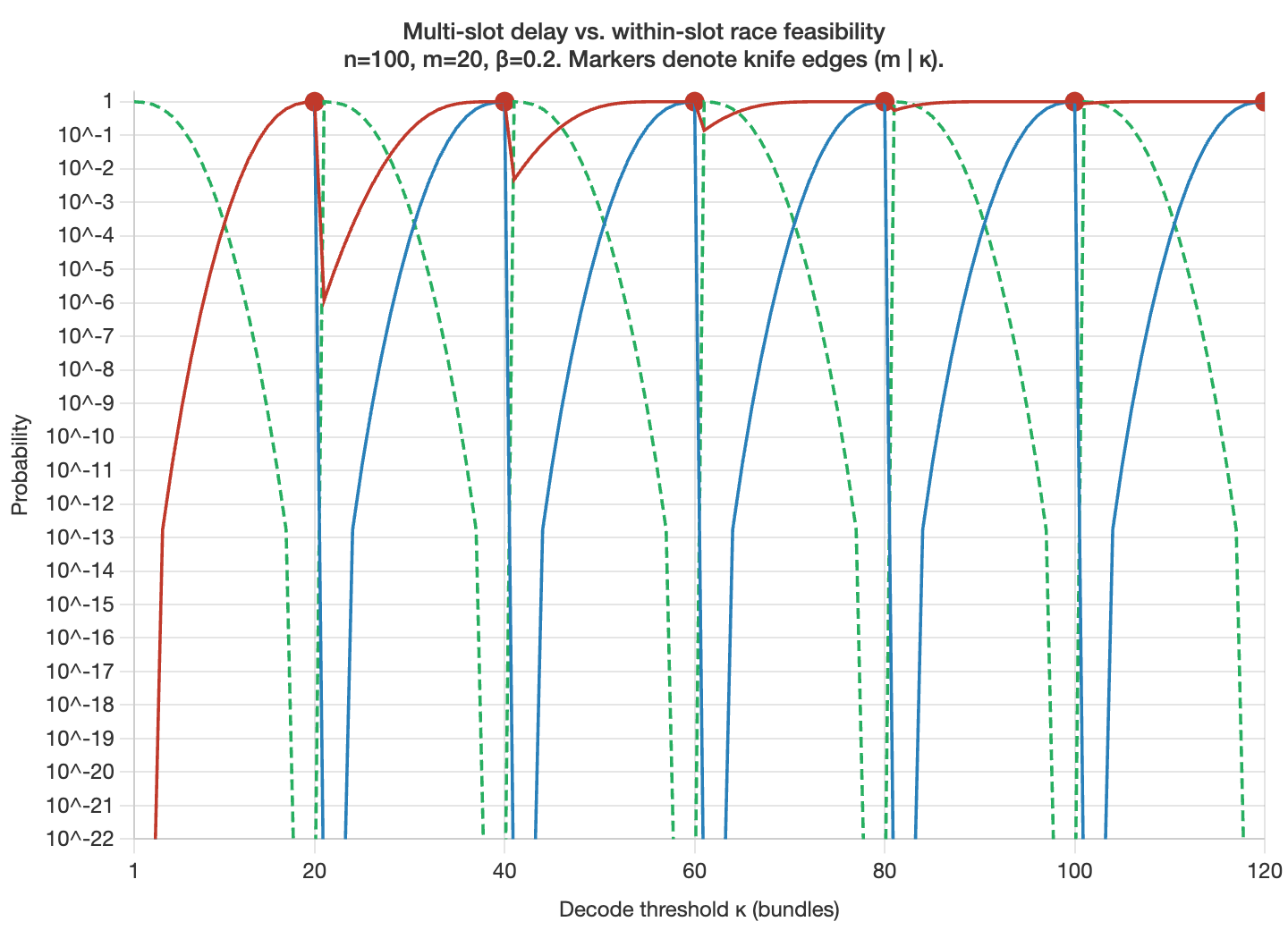}
\caption{Multi-slot delay risk under full withholding ($q^{(0)}$, red),
  ratchet-corrected delay ($q_{\mathrm{rat}}$, blue), and within-slot
  race feasibility ($q_{\mathrm{micro}}$, green dashed).  The three
  risks are complementary: $q^{(0)}$ peaks where $q_{\mathrm{micro}}$
  is negligible (knife edges, $\Delta = 0$) and vice versa.  The
  ratchet flattens the multi-slot delay envelope for $t^\star \ge 2$
  to the single-slot tail level, except on knife edges where it
  coincides with $q^{(0)}$.  No single $\kappa$ eliminates all three
  risks simultaneously; \textsf{PIVOT-$K$} and adaptive sending address
  the multi-slot component, while protocol-level timing addresses the
  within-slot component.
  Parameters: $n = 100$, $m = 20$, $\beta = 0.2$.}
\label{fig:delay_vs_micro}
\end{figure}
% \begin{figure}[ht]
% \centering
% \includegraphics[width=0.85\textwidth]{fig_delay_vs_micro}
% \caption{Multi-slot delay risk ($q^{(0)}$, red) vs.\ within slot
%   race feasibility ($q_{\mathrm{micro}}$, green) under a static
%   sender.  The two risks are complementary: $q^{(0)}$ peaks where
%   $q_{\mathrm{micro}}$ is negligible (knife edges, $\Delta = 0$)
%   and vice versa.  No single $\kappa$ eliminates both risks,
%   mechanism design (\textsf{PIVOT-$K$}) and adaptive sending address
%   the multi-slot component, while protocol level timing addresses
%   the within slot component.}
% \label{fig:delay_vs_micro}
% \end{figure}

\paragraph{Coalition decomposition.}
We now evaluate the coalition sufficient
condition~\eqref{eq:coalition_budget} from
Section~\ref{sec:coalition}.  Recall that for $\Delta > 0$, any
delay inducing coalition must withhold at least $\Delta + 1$ bundles,
and the aggregate direct revenue sacrifice must exceed the total
extractable MEV. Rearranging~\eqref{eq:coalition_B}, the coalition
sufficient bounty is
\[
  B_{\mathrm{coal}}
  \;:=\;
  \kappa\,\max\!\bigl(0,\;
    \alpha v\,\gamma^{t^\star} - (\Delta+1)f
  \bigr),
\]
where $\max(0, \cdot)$ reflects that no bounty is needed when fees
alone cover the deterrence. This evaluates to $880$ at
$\kappa = 10$, to $2{,}610$ at $\kappa = 30$, and to $4{,}302$ at
$\kappa = 50$.  In normalized fee units under
$f = 1 \leftrightarrow \$0.10$, these correspond to \$88, \$261, and
\$430.

These values are far above the monolithic static proxy
$B^\star_{\mathrm{static}} \approx (\alpha v/\beta)\,q^{(0)}$ from
Table~\ref{tab:main}. This is not a contradiction, the two quantities
answer different questions.  The monolithic proxy
$B^\star_{\mathrm{static}}$ is an expected value bound, it averages
over all sample paths, including the vast majority where the cartel is
not contacted enough times to mount the attack in the first place.
The coalition bound $B_{\mathrm{coal}}$ is a worst case, conditional
guarantee.  It must hold on the specific (rare) sample paths where at
least $\Delta + 1$ cartel lanes are contacted, ensuring that even on
those paths no internal sharing rule can sustain the coalition.  For
operating guidance we therefore use the monolithic static proxy and the
ratchet corrected proxy as the quantities that inform day to day
bounty pricing, and we treat $B_{\mathrm{coal}}$ as a conservative
stability certificate rather than a pricing rule.

Table~\ref{tab:coalition} reports the coalition decomposition.  The
column ``Unil.\ safe'' indicates whether a single withheld bundle is
strictly unprofitable (yes whenever $\Delta > 0$).  ``Equal-share''
reports the benchmark quantity $\alpha v \cdot \gamma^{t^\star}/(\Delta+1)$ obtained by splitting the transaction level MEV evenly across the minimal $\Delta+1$ withheld bundle opportunities required for delay. Under the static sender these opportunities need not correspond to $\Delta+1$ distinct
lanes, under the ratchet they do, because each withholding permanently
burns the corresponding lane. The column $B_{\mathrm{coal}}$ reports the coalition sufficient bounty defined above. $B^\star_{\mathrm{static}}$ is the monolithic static proxy from Table~\ref{tab:main}. For knife edge rows ($\Delta = 0$), the coalition decomposition does not apply (the attack is unilateral) and $B_{\mathrm{coal}}$ is not defined.

\begin{table}[ht]
\centering
\caption{Coalition decomposition
  ($n{=}100$, $m{=}20$, $\beta{=}0.2$, $\alpha v{=}100$,
   $\gamma{=}0.99$, $f{=}1$;
   unit: $f{=}1 \leftrightarrow \$0.10$).}
\label{tab:coalition}
\begin{tabular}{@{}rrrcrr@{}}
\toprule
$\kappa$ & $\Delta$ &
  Unil.\ safe &
  Equal share &
  $B_{\mathrm{coal}}$ &
  $B^\star_{\mathrm{static}}$ \\
\midrule
10  & 10 & yes & $9.0$  & $880$  & $0.04$ \\
20  & 0  & \textbf{no}  & $99.0$ & ---   & $497$ \\
30  & 10 & yes & $8.9$  & $2{,}610$ & $68$ \\
50  & 10 & yes & $8.8$  & $4{,}302$ & $350$ \\
100 & 0  & \textbf{no}  & $95.1$ & ---   & $500$ \\
\bottomrule
\end{tabular}
\end{table}

Two features are worth highlighting.  First, for every row with
$\Delta > 0$, no individual lane can profitably withhold alone,
regardless of the bounty level.  Second, the coalition bound
$B_{\mathrm{coal}}$ exceeds the monolithic proxy
$B^\star_{\mathrm{static}}$ in every $\Delta > 0$ row, reflecting the
gap between an expected value pricing rule and a worst case conditional
stability certificate.

Figure~\ref{fig:coalition_share} extends this comparison across
$\kappa \in [1, 120]$.  Panel~(a) plots the per-member MEV share
against the per-lane opportunity cost~$f$, with three regimes shaded:
fees alone break the coalition (blue), a moderate bounty suffices
(green), and knife edges where the monolithic IC binds (red).
Panel~(b) plots $B_{\mathrm{coal}}$ alongside
$B^\star_{\mathrm{static}}$.

\begin{figure}[ht]
\centering
\includegraphics[width=0.85\textwidth]{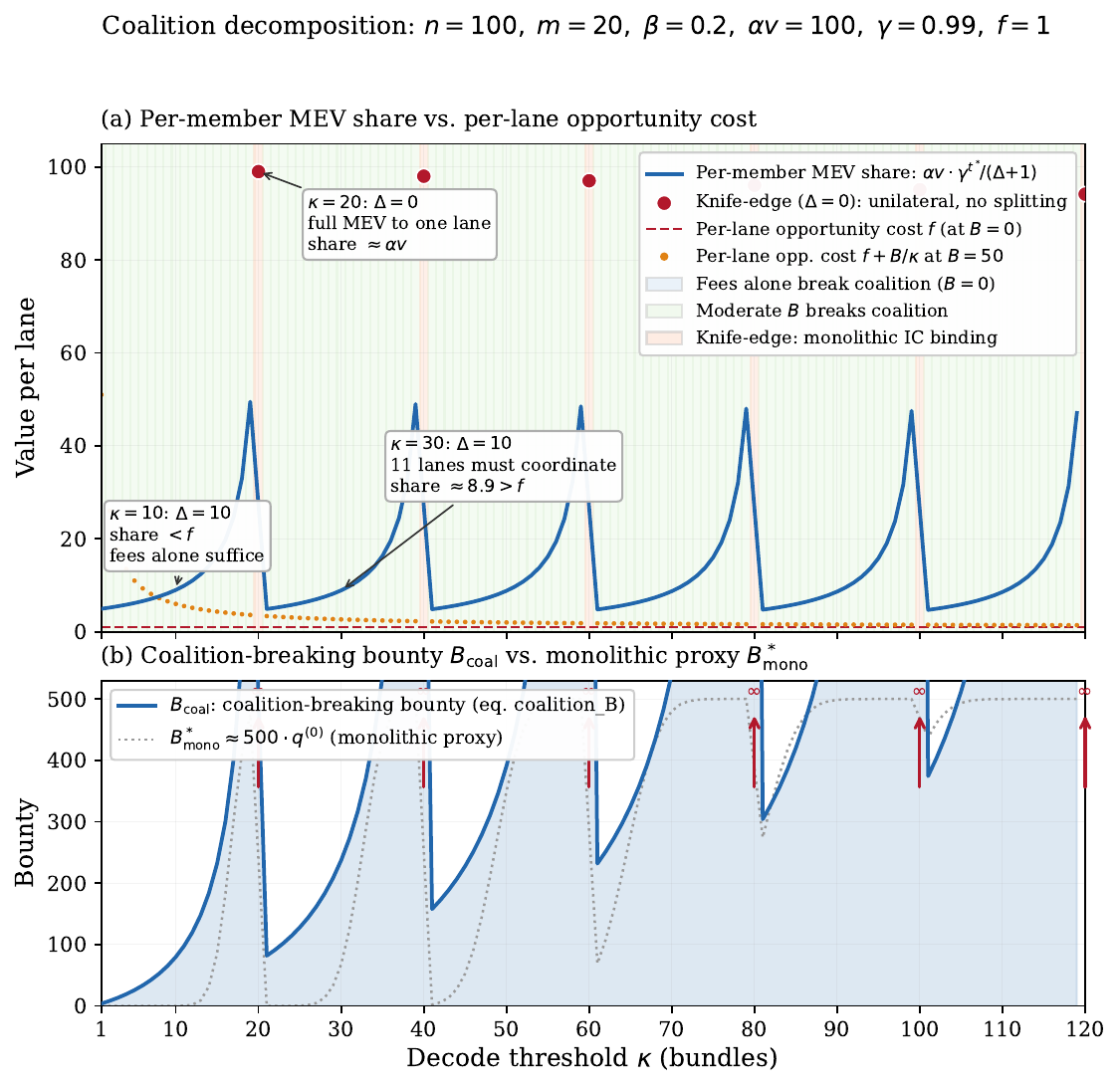}
\caption{Coalition decomposition across $\kappa$.
  Panel~(a): equal share benchmark
  $\alpha v \cdot \gamma^{t^\star}/(\Delta+1)$ (blue) versus
  per-lane opportunity cost~$f$ (red dashed).  Shaded regions:
  fees alone break the coalition (blue), moderate bounty suffices
  (green), monolithic IC binds (red).  Dots mark knife edges.
  Panel~(b): coalition bounty $B_{\mathrm{coal}}$ (solid
  blue) versus monolithic proxy (dotted grey). Red arrows mark
  knife edges where $B_{\mathrm{coal}}$ is undefined.
  Parameters: $n = 100$, $m = 20$, $\beta = 0.2$,
  $\alpha v = 100$, $\gamma = 0.99$, $f = 1$.}
\label{fig:coalition_share}
\end{figure}

The practical interpretation is layered.  For $\Delta > 0$, full inclusion is protected against unilateral withholding and, when \eqref{eq:coalition_budget} holds, against delay inducing coalition deviations.  For knife edge rows ($\Delta = 0$), the attack is unilateral, the full monolithic IC binds, and large bounties
(the knife edge bound~\eqref{eq:dynIC_delta0_B}) or parameter
adjustment to avoid $\Delta = 0$ remain necessary.  Under the ratchet
with $\kappa = 30$, each withholder permanently loses its lane
(Section~\ref{sec:adv_ratchet}), and the withholding equilibrium is
eliminated by fees alone.

\paragraph{Dynamic IC on the knife edge.}
For knife edge rows ($\Delta = 0$), the knife edge
bound~\eqref{eq:dynIC_delta0_B} yields a closed form sufficient
condition.  With $\gamma = 0.99$ and $f = 1$ ($\$0.10$),
equation~\eqref{eq:dynIC_delta0_B} gives $B \gtrsim 2{,}500$
($\$250$) for $\kappa = 20$ ($t^\star = 1$) and
$B \gtrsim 10{,}000$ ($\$1{,}000$) for $\kappa = 100$
($t^\star = 5$).  These large bounties reflect the severity of the
knife edge threat under the static sender.

Table~\ref{tab:cost} translates the bounty proxy into USD across three
MEV tiers, using the operating point $\kappa = 30$ ($t^\star = 2$,
$\Delta = 10$) where the ratchet is active.

\begin{table}[ht]
\centering
\caption{Illustrative sender bounty cost in USD ($\kappa = 30$,
$t^\star = 2$, $\Delta = 10$, $f = \$0.10$).}
\label{tab:cost}
\begin{tabular}{@{}lrrrr@{}}
\toprule
MEV tier & $\alpha v$ &
  $B^\star_{\mathrm{static}}$ &
  $B^\star_{\mathrm{ratchet}}$ &
  $B^\star_{\mathrm{ratchet}} / \alpha v$ \\
\midrule
Routine swap     & \$5    & \$3.40   & \$0.002  & $0.04\%$ \\
Sandwich / arb   & \$50   & \$34     & \$0.02   & $0.04\%$ \\
Liquidation      & \$5{,}000 & \$3{,}400 & \$2.00 & $0.04\%$ \\
\bottomrule
\end{tabular}
\end{table}

Under the ratchet, the required bounty is $0.04\%$ of the protected
MEV across all three tiers.  For the middle tier this is roughly
\$0.02 to protect a transaction with $\alpha v = \$50$; for the
largest tier it is roughly \$2 to protect $\alpha v = \$5{,}000$.  By
contrast, the static sender proxy at the same operating point is about
$68\%$ of $\alpha v$, which is why the adaptive sender is the
operationally decisive lever for making the bounty economically viable.

\paragraph{Operational guidance.}
For senders targeting $t^\star = 1$, the ratchet cannot help within
the transaction itself, so the right lever is to choose $\kappa < m$
and keep positive slack.  For senders with multi-slot decode horizons
($t^\star \ge 2$), the ratchet collapses the relevant delay event to a
single slot tail and drives the required bounty down by orders of
magnitude on every positive slack instance.  Knife edge instances
($m \mid \kappa$) remain the configurations to avoid, the attack is
unilateral, the knife edge IC bound governs, and either a large bounty
or a different parameter choice is needed.  For $t^\star \ge 2$,
Sedna's until decode privacy further strengthens this conclusion
because the adversary's first slot withholding decision is made before
later slot evidence reveals how much slack is actually available.

\section{Conclusion}
\label{sec:conclusion}

Sedna's coded multi-proposer dissemination improves goodput over
single-leader designs and provides a meaningful degree of privacy and
censorship resistance through its until decode hiding property. These
guarantees, however, are not absolute. Sedna alone does not fully
protect the sender against arbitrary cartel sizes, and the withholding
attack analysed in this paper shows that the gap is real. We show that a combination of three complementary design choices contains this threat under explicit, computable conditions.

\textsf{PIVOT-$K$} concentrates bounty reward on the pivotal bundles
that trigger decoding, raising the per-lane opportunity cost of
withholding and transferring the reward to a competitor whenever a
bundle is withheld.  The adaptive sender ratchet detects and
permanently excludes non-responsive lanes after each slot, collapsing
what would be a multi-slot accumulation attack into a single slot
deficit for any transaction with $t^\star \ge 2$.  The ratchet's power
is amplified by a structural informational asymmetry: the adversary
does not know the slack $\Delta$ (and hence cannot compute how many
lanes it must sacrifice) before committing to its first slot
withholding, because $\kappa$ is hidden behind a computationally hiding
commitment until decode.  This forces the cartel to gamble blindly, and
every miscalibrated sacrifice permanently shrinks its effective network
fraction.  The coalition decomposition then shows that for $\Delta > 0$,
unilateral withholding is already strictly unprofitable: any successful
delay requires coordinating at least $\Delta + 1$ withheld bundles, and
the aggregate direct revenue sacrifice of any such coalition exceeds the
total extractable MEV from this transaction.

Two residuals remain. Knife edge instances ($\Delta = 0$) admit
unilateral attacks and require large bounties or parameter adjustment,
avoiding $m \mid \kappa$ is the primary practical recommendation.
Within slot decode races are irreducible by any finalization based
incentive mechanism but exponentially unlikely in~$m$, the number of
contacted lanes, for $\Delta > 0$.

All bounds in this paper are derived under the conservative assumption
that the adversary knows $\Delta$ exactly.  In the illustrative
calibration of Section~\ref{sec:numerics}, the ratchet backed bounty
under this worst case is about $0.04\%$ of the protected MEV, roughly
\$2 for a transaction with $\alpha v = \$5{,}000$ and roughly \$200 for
a transaction with $\alpha v = \$500{,}000$.  Without this assumption,
the adversary must commit its withholding before learning the slack.
A cartel that overestimates $\Delta$ withholds too few bundles and
wastes its sacrifice; one that underestimates it burns lanes for zero
MEV gain.  Either way, the effective required bounty under genuine
uncertainty is strictly lower than the figures above. Quantifying the
optimal Bayesian withholding strategy under a prior over~$\kappa$ is a
natural extension that would tighten these bounds further, particularly
on the instances where the current analysis is most pessimistic.

Sedna's incentive layer, as analysed here, makes coded multi-proposer
dissemination economically robust at modest cost and, to our knowledge,
provides the first formal mechanism design treatment of incentive compatible coded dissemination in the MCP setting. For senders seeking even stronger protection, combining Sedna with a lightweight commit reveal scheme at the slot boundary~\cite{shutter} is expected to eliminate within slot decode races as well, closing the residual gap identified in Section~\ref{sec:micro} without the complexity or expense of full threshold encryption.

\section{Acknowledgments}
We would like to thank Paolo Serafino, Steven Landers, Alberto Sonnino, Lefteris Kokoris-Kogias, and Mahimna Kelkar for their conversations related to this work.

\printbibliography
    
\appendix
    
\section{Ceiling Effects, Rounding, and Overshoot}
\label{app:ceilings}
    
    The main text adopts the convention $K = \kappa s$ so that each of the first $\kappa$ included bundles contributes exactly $s$ pivotal indices. In practice, $K$ need not be a multiple of $s$, this appendix records the exact setting. Fix $K\ge 1$ and $s\ge 1$ and define
    \[
    \kappa := \left\lceil\frac{K}{s}\right\rceil,
    \qquad
    r := K-(\kappa-1)s \in \{1,2,\dots,s\}.
    \]
    Thus $K=(\kappa-1)s+r$, and decoding becomes possible after $\kappa$ included bundles, with the last bundle contributing only $r$ pivotal indices.
    
    \begin{lemma}[Pivotal indices with non-divisible $(K,s)$]
    \label{lem:pivot_nondv}
    Assume each bundle carries $s$ fresh indices (no repeats within or across bundles) and Sedna resolves indices deterministically. Then the pivotal index set $\mathcal{I}^\star$ consists of all $s$ indices from each of the first $\kappa-1$ included bundles, and exactly $r$ indices from the $\kappa$-th included bundle (the first $r$ indices of that bundle in the canonical scan).
    \end{lemma}
    
    \begin{corollary}[Per-bundle payments under \textsf{PIVOT-$K$} without divisibility]
    \label{cor:per_bundle_nondv}
    Under \textsf{PIVOT-$K$}, each pivotal index pays $B/K$. Hence each of the first $\kappa-1$ included bundles earns $\frac{s}{K}B$, while the $\kappa$-th included bundle earns $\frac{r}{K}B$. In particular, the total bounty paid is exactly $B$ and the ``uniform $B/\kappa$ per pivotal bundle'' view is recovered only when $r=s$.
    \end{corollary}
    
    All martingale arguments used to approximate cartel bounty share extend by treating the payoff as a sum over pivotal indices  rather than pivotal bundles. Since only the final bundle is partially pivotal, the non-divisibility effect is bounded by one bundle worth of indices; in particular, any additive error term that scales as $O(m/\kappa)$ in the divisible case remains of order $O(m/\kappa)$ after rounding, with an additional $O(1/K)$ term coming from the final partial bundle.
    
    \section{Dynamic Withholding Policies}
    \label{app:dynamic}
The main text analyzes stationary partial withholding policies
parameterized by $w\in[0,1]$. A fully general deviation allows the
cartel to choose, in each slot, which of its contacted lanes will
include, possibly as a function of the full history and within slot
information. This appendix records the general policy notation and one
exact monotonicity fact used in the main text, under \textsf{PIVOT-$K$}, the cartel's pivotal bundle count is pathwise nondecreasing in inclusion. We do not claim a policy invariant closed form for the dynamic MEV term, the exact worst case delay benchmark needed in the main text is Theorem~\ref{thm:worst_w0_rigorous}.
    
    \subsection{General policy model}
    Fix a transaction and define per slot cartel contacts $A_t$ and honest contacts $H_t=m-A_t$ as in the main model. A general cartel policy $\pi$ chooses, in each slot $t$, a number $X_t^\pi\in\{0,1,\dots,A_t\}$ of cartel bundles to include on-chain. Let
    \[
    S_t := \sum_{u=1}^t A_u,
    \qquad
    U_t^\pi := \sum_{u=1}^t (H_u+X_u^\pi)= tm - S_t + \sum_{u=1}^t X_u^\pi.
    \]
    Define the on-chain inclusion time
    \[
    T^\pi := \inf\{t: U_t^\pi\ge \kappa\},
    \qquad
    t^\star := T(1)=\left\lceil\frac{\kappa}{m}\right\rceil,
    \qquad
    q^\pi := \Prb[T^\pi>t^\star].
    \]
    Under the corrected information set, the discounted multi-slot MEV term is
    \[
    G^\pi := \E\!\left[\gamma^{t^\star}\1\{T^\pi>t^\star\}\right]=\gamma^{t^\star}q^\pi.
    \]
    
    \subsection{Pathwise monotonicity of pivotal bundle count under \textsf{PIVOT-$K$}}
    \textsf{PIVOT-$K$} pays $B/\kappa$ to each of the first $\kappa$ included bundles. Thus the cartel bounty is proportional to the number of cartel bundles among the first $\kappa$ included bundles. Fix a realization of the lane contact process and a deterministic total order of included bundles. Any policy $\pi$ induces a length $\kappa$ prefix of included bundles at the moment decoding becomes possible.
    
    \begin{lemma}[Adding cartel bundles cannot reduce cartel pivotal bundle count]
    \label{lem:pathwise_mono_pivot}
    Fix a sample path and a deterministic ordering. Let $\pi$ and $\pi'$ be two policies such that, up to the first time $\kappa$ bundles are included under $\pi'$, policy $\pi'$ includes every cartel bundle that $\pi$ includes, and possibly additional cartel bundles (i.e.\ $\pi'$ only differs by inserting cartel bundles earlier). Let $J_\kappa(\pi)$ (resp.\ $J_\kappa(\pi')$) be the number of cartel bundles among the first $\kappa$ included bundles under $\pi$ (resp.\ $\pi'$).
    Then
    \[
    J_\kappa(\pi') \;\ge\; J_\kappa(\pi)\quad\text{(pathwise)}
    \]
    \end{lemma}
    
    \begin{proof}
    Consider the first $\kappa$ included bundles under $\pi$ as a length $\kappa$ sequence over $\{H,C\}$ (honest/cartel). Policy $\pi'$ inserts additional $C$ symbols earlier and then truncates to the first $\kappa$ elements. Each insertion either replaces an $H$ at the end of the prefix (increasing $J_\kappa$ by $1$) or replaces a $C$ (leaving $J_\kappa$ unchanged). Thus the cartel count in the prefix cannot decrease.
    \end{proof}
    
    Theorem~\ref{thm:worst_w0_rigorous} captures the corrected monotonicity geometry where delay risk is maximized by full withholding, while direct revenue terms (fees and \textsf{PIVOT-$K$}) are monotone in inclusion. Quantifying optimal dynamic deviations therefore reduces to balancing a delay induced MEV option against immediate fee and bounty revenue, rather than relying on a policy invariant $e^{-\kappa\KL(1/2\|\beta)}$ suppression. A complete dynamic incentive constraint still requires solving the cartel's MDP (or restricting the deviation class) to quantify the optimal tradeoff between a smaller direct revenue stream and a larger (but exponentially bounded) MEV option.
    
    \subsection{Minimal sabotage under general fee markets}
    \label{app:fee_market}
    
    Theorem~\ref{thm:dyn_best_binary} assumes static per-bundle payments
    (\S\ref{sec:model}), which guarantees nonnegative marginal direct
    payoff from including an additional admissible bundle.  This
    subsection records a counterexample showing the result can fail
    without that assumption, and states a sufficient condition for the
    general case.
    
    \begin{lemma}[Minimal sabotage can fail with negative net
      marginal inclusion payoff]
    \label{prop:fee_counterexample}
    There exist fee environments in which, even when the MEV term depends
    only on the binary delay event $\{T^\pi > t^\star\}$, an optimal
    deviation that achieves delay withholds strictly more than
    $\Delta + 1$ bundles.
    \end{lemma}
    
    \begin{proof}
    Consider a single-slot instance with $t^\star = 1$, $\kappa = m$,
    so $\Delta = 0$ and delay requires withholding at least one bundle.
    Suppose the lane has capacity for $m$~bundles and faces an
    alternative transaction worth $R$ per bundle in the same slot,
    with~$R$ so large that including a Sedna bundle displaces an
    alternative bundle at strictly negative net marginal payoff.
    Conditional on achieving delay, the cartel strictly prefers to
    withhold as many Sedna bundles as possible to free capacity for the
    alternative transactions.
    \end{proof}
    
    A sufficient condition to restore minimal sabotage is that the net
    marginal payoff from including an additional admissible Sedna bundle
    is nonnegative after accounting for opportunity cost.
    
    \begin{assumption}[Nonnegative net marginal inclusion payoff]
    \label{as:net_nonneg}
    For every slot $t \le t^\star$ and every history consistent with the
    protocol, the lane's net marginal payoff from including one
    additional admissible bundle for this transaction is nonnegative
    (after accounting for displaced fee revenue, congestion pricing, or
    other fee-market externalities).
    \end{assumption}
    
    \begin{theorem}[Minimal sabotage under general fee markets]
    \label{thm:dyn_best_binary_general_fees}
    Under Assumption~\ref{as:net_nonneg} and \textsf{PIVOT-$K$}, the
    conclusion of Theorem~\ref{thm:dyn_best_binary} holds: conditional
    on achieving $T^\pi > t^\star$, an optimal policy withholds exactly
    $\Delta + 1$ cartel-addressed bundles and includes all others.
    \end{theorem}
    
    \begin{proof}
    Delay requires at least
    $\Delta + 1$ withholdings.  Under
    Assumption~\ref{as:net_nonneg}, each additional inclusion weakly
    increases fee payoff.  Under \textsf{PIVOT-$K$}, it cannot reduce
    the cartel's pivotal-bundle count
    (Lemma~\ref{lem:pathwise_mono_pivot}).  Withholding beyond
    $\Delta + 1$ does not increase the MEV term.  Hence an optimal
    policy withholds exactly $\Delta + 1$ bundles.
    \end{proof}
    
    Assumption~\ref{as:net_nonneg} holds if Sedna bundles occupy a
    reserved bandwidth slice (opportunity cost $\approx 0$), or if the
    sender's lane-bound ticket price dominates the lane's marginal
    opportunity cost.
    
    \subsection{Bandwidth knapsack}
    \label{app:knapsack}
    
    The cardinality budget $\sum_i x_i\le J$ is a clean toy model but often the wrong bottleneck. A more realistic coupling is a private capacity constraint that the cartel can withhold, transport, and process only so many bundles (or symbols) for private decoding per unit time. Let $c_i\ge 0$ be the cartel's private capacity cost of attacking transaction $i$ (e.g., the expected number of withheld bundles needed to ensure a delay event, which is at least $\Delta_i+1$ on knife edge instances), and let $\Delta_i$ be the cartel's incremental gain from attacking transaction $i$ relative to full inclusion, e.g.
    \[
    \Delta_i := \sup_{\pi_i\in\mathcal{A}_i}\Big(\Pi_i(\pi_i)-\Pi_i(\text{include})\Big),
    \]
    where $\mathcal{A}_i$ is the allowed attack set for transaction $i$. If the cartel has total capacity $C$, the selection problem becomes the $0$-$1$ knapsack:
    \[
    \max_{x\in\{0,1\}^N}\ \sum_{i=1}^N x_i\,\Delta_i
    \quad\text{s.t.}\quad
    \sum_{i=1}^N x_i\,c_i \le C.
    \]
    Unlike top $J$, the knapsack constraint can make the cartel prefer fewer large gain attacks or many small ones depending on the $\Delta_i/c_i$ ratios. Mechanisms like \textsf{PIVOT-$K$} reduce $\Delta_i$ transaction by transaction. Under capacity coupling, even a moderate reduction in $\Delta_i$ can push a transaction out of the optimal attack set because it changes the knapsack ranking. This is the setting where ``cross transaction deterrence'' is economically meaningful.
    
    \section{Bayesian Bounty Choice}
    \label{app:bayes}
    
    Users may not know the effective cartel fraction $\beta$ (or the mapping from early decode to realized MEV loss).
    A simple abstraction treats the withholding proof threshold bounty as a random variable induced by uncertainty.
    
    \begin{assumption}[Threshold response under uncertainty]
    \label{as:threshold_uncertainty}
    Given environment parameter $\omega$ (e.g.\ $\beta$, latency parameters, fee market conditions), there exists a withholding proof threshold $B^\star(\omega)$ such that the cartel includes iff $B\ge B^\star(\omega)$. The user has a prior $\omega\sim\mathcal{P}$ and hence an induced distribution for $B^\star$ with CDF $F$.
    \end{assumption}
    
    Let $U_{\mathrm{inc}}$ denote the user's expected utility if inclusion is induced (net of bandwidth fees, excluding the bounty transfer), and $U_{\mathrm{wh}}$ the user's expected utility under withholding.
    
    \begin{lemma}[Optimal posted bounty under a prior over $B^\star$]
    \label{prop:bayes_appendix}
    Under Assumption~\ref{as:threshold_uncertainty}, the user's expected utility from posting bounty $B$ is
    \[
    \mathcal{U}(B)
    =
    F(B)\cdot (U_{\mathrm{inc}}-B) + (1-F(B))\cdot U_{\mathrm{wh}}.
    \]
    If $F$ has density $f$ and an interior optimum exists, it satisfies the first order condition
    \[
    f(B)\cdot (U_{\mathrm{inc}}-U_{\mathrm{wh}}-B) \;=\; F(B)
    \]
    \end{lemma}
    \begin{proof}
From the user's perspective, the threshold $B^\star$ is a random variable with CDF $F$. If the user posts a bounty $B \ge B^\star$, the cartel is incentivized to choose full inclusion. In this event, the user pays the bounty $B$ and receives utility $U_{\mathrm{inc}} - B$. Conversely, if $B < B^\star$, the cartel withholds, and the user receives the default delayed utility $U_{\mathrm{wh}}$. 

Taking the expectation over the distribution of $B^\star$ yields the user's expected utility:
\[
\mathcal{U}(B) = \Prb(B \ge B^\star)(U_{\mathrm{inc}} - B) + \Prb(B < B^\star)U_{\mathrm{wh}} = F(B)(U_{\mathrm{inc}} - B) + (1-F(B))U_{\mathrm{wh}}.
\]
Assuming the distribution $F$ is continuous and differentiable with density $f$, we find the optimal posted bounty by differentiating $\mathcal{U}(B)$ with respect to $B$:
\[
\frac{d\mathcal{U}}{dB} = f(B)(U_{\mathrm{inc}} - B) - F(B) - f(B)U_{\mathrm{wh}}.
\]
Setting the derivative to zero for an interior maximum and rearranging the terms yields the first-order condition:
\[
f(B)(U_{\mathrm{inc}} - U_{\mathrm{wh}} - B) \;=\; F(B).
\]
\end{proof}
    
    \section{Adaptive sender policies and recovery}
    \label{sec:adaptive_sender}
    The main model fixes $(m,s)$ across slots and samples lanes uniformly each slot. Operationally, a sender can adapt based on observed on-chain accumulation, e.g.\ by increasing $m_t$, changing $s_t$, or biasing lane sampling away from suspected withholders. We record the minimal extension needed to preserve the delay geometry, a full equilibrium analysis with adaptive senders is left for future work. Fix a deterministic schedule $\{m_t\}_{t\ge 1}$ with $m_t\in\{0,1,\dots,n\}$ and keep $s$ fixed so that $\kappa=\lceil K/s\rceil$ bundles are needed. In slot $t$, the sender samples $m_t$ lanes uniformly without replacement. Let
    \[
    A_t\sim \mathrm{Hypergeom}(n,\beta n,m_t),\qquad
    S_t:=\sum_{u=1}^t A_u,\qquad
    M_t:=\sum_{u=1}^t m_u.
    \]
    Under $w$ withholding, the on-chain bundle count is
    \[
    U_t(w)=M_t-(1-w)S_t.
    \]
    Define the honest horizon and slack by
    \[
    t^\star := \inf\{t:M_t\ge \kappa\},
    \qquad
    \Delta := M_{t^\star}-\kappa \in\{0,1,\dots,m_{t^\star}-1\}.
    \]
    
    \begin{lemma}[Delay event under a time varying contact schedule]
    \label{lem:threshold_w_schedule}
    For any $w\in[0,1]$,
    \[
    \Prb[T(w)>t^\star]
    =
    \Prb\!\left[U_{t^\star}(w)<\kappa\right]
    =
    \Prb\!\left[(1-w)S_{t^\star}>\Delta\right].
    \]
    \end{lemma}
    
    Unlike the fixed $m$ design, the sender can enforce $\Delta>0$ even at small $t^\star$ by over contacting in the last pre horizon slot. For example, if targeting $t^\star=1$, choosing $\kappa<m_1$ creates slack $\Delta=m_1-\kappa$ and strictly reduces the full withholding delay probability $q^{(0)}=\Prb[S_1>\Delta]$. If the sender forms an effective per slot cartel hit rate $\beta_t$ (e.g.\ from long run lane behavior), then all KL type tail bounds extend by replacing $\beta$ with an upper envelope $\sup_t \beta_t$.

    \subsection{Honest redemption failures and graceful degradation}
    \label{app:honest_miss}
    
    The main text assumes perfect honest redemption: every honest lane
    that receives a valid ticket includes the corresponding bundle.  In
    practice, honest lanes may fail to redeem due to network jitter or
    downtime.  This subsection models such failures and shows that the
    ratchet's delay bound degrades gracefully.
    
    Binary exclusion is best viewed as sender-side routing: a lane that
    fails to redeem its ticket is deemed unreliable for this transaction
    and excluded from future contacts for this~$\id$.  This does not
    require protocol-level slashing or perfect attribution of malice.
    
    \begin{assumption}[Honest ticket redemption reliability]
    \label{as:honest_fail}
    When an honest lane is contacted with a valid lane-bound ticket for
    this transaction, it redeems (includes an admissible occurrence)
    with probability at least $1 - \varepsilon$, independently across
    contacts, for some $\varepsilon \in [0,1)$.
    \end{assumption}
    
    Fix a horizon $t^\star$ and planned total contacts
    \[
      M_{t^\star} := m + \sum_{t=2}^{t^\star} m_t,
      \qquad
      \Delta_{\mathrm{rec}} := M_{t^\star} - \kappa.
    \]
    Let $W$ be the total number of strategic cartel non-redemptions up
    to time~$t^\star$ for this transaction, and let $E$ be the total
    number of honest non-redemptions among all honest contacts up
    to~$t^\star$.
    
    \begin{lemma}[Delay bound with honest misses]
    \label{prop:ratchet_epsilon}
    Under Assumption~\ref{as:honest_fail}, for any adversary strategy
    and any sender schedule $\{m_t\}$, the delay probability satisfies
    \[
      \Prb[T > t^\star]
      \;\le\;
      \Prb\bigl[W + E > \Delta_{\mathrm{rec}}\bigr].
    \]
    Moreover, $E$ is stochastically dominated by a binomial:
    $\Prb[E \ge x]
    \le \Prb[\mathrm{Bin}(M_{t^\star}, \varepsilon) \ge x]$
    for every integer $x \ge 0$, and therefore
    \[
      \Prb[T > t^\star]
      \;\le\;
      \Prb\!\left[\,
        W + \mathrm{Bin}(M_{t^\star}, \varepsilon)
        > \Delta_{\mathrm{rec}}
      \,\right].
    \]
    \end{lemma}
    
    \begin{proof}
    By time~$t^\star$, at least $M_{t^\star} - (W + E)$ admissible
    bundles are finalised.  Inclusion fails iff
    $M_{t^\star} - (W + E) < \kappa$, i.e.\
    $W + E > \Delta_{\mathrm{rec}}$.  For the binomial domination, $E$
    is a sum of failure indicators with per-contact probability at
    most~$\varepsilon$, and binomial upper tails are increasing in the
    number of trials.
    \end{proof}
    
    Under the one-shot deviation model of \S\ref{sec:firstslot} with
    maximal first-slot withholding $W = A_1$, the bound specialises to
    \[
      \Prb[T > t^\star]
      \;\le\;
      \Prb\!\left[\,
        A_1 + \mathrm{Bin}(M_{t^\star}, \varepsilon)
        > \Delta_{\mathrm{rec}}
      \,\right],
    \]
    which reduces to the perfect detection bound of
    \eqref{eq:q_rat} when $\varepsilon = 0$.

    \section{Sealing deadline instantiation of $\rho$}
    \label{app:micro_pipeline}
    
Fix a slot duration $\tau_{\mathrm{slot}}$ and a sealing deadline $\tau_{\mathrm{seal}} \in (0, \tau_{\mathrm{slot}}]$. The cartel needs a deterministic reaction time $\tau_0 := \tau_{\mathrm{dec}} + \tau_{\mathrm{act}}$, covering decode and constructing, signing, and propagating the MEV action. Conditional on $A_{t^\star} = a$, model the arrival times of the $a$~cartel addressed bundles as i.i.d.\ random variables $T_1, \dots, T_a$ on $[0, \tau_{\mathrm{seal}}]$ with CDF~$F(\cdot)$.  Let $T_{(r)}$ be the $r$-th order statistic. Within slot success requires $T_{(r)} + \tau_0 \le \tau_{\mathrm{seal}}$, giving
\begin{equation}
\label{eq:rho_deadline}
    \rho(a, r)
    \;=\;
    \Prb[T_{(r)} \le \tau_{\mathrm{seal}} - \tau_0]
    \;=\;
    \Prb[\mathrm{Bin}(a, p) \ge r],
    \qquad
    p := F(\tau_{\mathrm{seal}} - \tau_0).
\end{equation}
For example, if $F(t) = 1 - e^{-\lambda t}$ for $t \in [0, \tau_{\mathrm{seal}}]$, then $p = 1 - e^{-\lambda(\tau_{\mathrm{seal}} - \tau_0)}$ and \eqref{eq:rho_deadline} is a binomial tail with an explicit parameter.
\end{document}

%% file: preamble.tex
\usepackage{graphicx}
\usepackage{amsthm}
\usepackage{amssymb}
\usepackage{amsmath}

\newtheorem{lemma}{Lemma}
\usepackage{xcolor}
\usepackage{todonotes}
\usepackage[most]{tcolorbox}
\usepackage{graphicx}
\usepackage{placeins}
\usepackage[
backend=biber,
style=alphabetic,
sorting=ynt
]{biblatex}